\documentclass[pra,twocolumn,amsmath,amssymb,nofootinbib]{revtex4-2}

\usepackage{graphicx}
\usepackage{dcolumn}
\usepackage{braket}
\usepackage{setspace}
\usepackage{bm}
\usepackage{amsmath}
\usepackage{amsthm}
\usepackage{mathtools}
\usepackage{csquotes}
\usepackage{mathrsfs}
\usepackage{subfigure}
\usepackage[svgnames]{xcolor}
\usepackage[colorlinks,citecolor=red,urlcolor=red,linkcolor=red,bookmarks=false, hypertexnames=true]{hyperref}

\allowdisplaybreaks 

\usepackage{amsfonts}
\usepackage{amssymb}
\usepackage{cancel}

\newtheorem{theorem}{Theorem}
\newtheorem{remark}{Remark}
\newtheorem{corollary}{Corollary}
\newtheorem{assumption}{Assumption}

\usepackage{tikz}
\usetikzlibrary{quantikz}

\usepackage{xcolor}

\begin{document}
\title{{Thermodynamic Constraints on Quantum Information Gain and Error Correction:\\A Triple Trade-Off}}
\author{Arshag Danageozian}
\email{Corresponding Author: arshag.danageozian@gmail.com}
\affiliation{Hearne  Institute  for  Theoretical  Physics,  Department  of  Physics  and  Astronomy,and  Center  for  Computation  and  Technology,  Louisiana  State  University,  Baton  Rouge,  Louisiana  70803,  USA}
\author{Francesco Buscemi}
\affiliation{Department of Mathematical Informatics, Nagoya University, Furo-cho Chikusa-ku, Nagoya, 464-8601, Japan}
\author{Mark M. Wilde}
\affiliation{Hearne  Institute  for  Theoretical  Physics,  Department  of  Physics  and  Astronomy,and  Center  for  Computation  and  Technology,  Louisiana  State  University,  Baton  Rouge,  Louisiana  70803,  USA}

\date{\today}
\begin{abstract}
Quantum error correction (QEC) is a procedure by which the quantum state of a system is protected against a known type of noise, by preemptively adding redundancy to that state. Such a procedure is commonly used in quantum computing when thermal noise is present. Interestingly, thermal noise has also been known to play a central role in quantum thermodynamics (QTD). This fact hints at the applicability of certain QTD statements in the QEC of thermal noise, which has been discussed previously in the context of Maxwell's demon. In this article, we view QEC as a quantum heat engine with a feedback controller (i.e., a demon). We derive an upper bound on the measurement heat dissipated during the error-identification stage in terms of the Groenewold information gain, thereby providing the latter with a physical meaning also when it is negative. Further, we derive the second law of thermodynamics in the context of this QEC engine, operating with general quantum measurements. Finally, we show that, under a set of physically motivated assumptions, this leads to a fundamental triple trade-off relation, which implies a trade-off between the maximum achievable fidelity of QEC and the super-Carnot efficiency that heat engines with feedback controllers have been known to possess. A similar trade-off relation occurs for the thermodynamic efficiency of the QEC engine and the efficacy of the quantum measurement used for error identification.
\end{abstract}

\maketitle
\tableofcontents

\section{Introduction}

The link between quantum thermodynamics (QTD) and information theory has been flourishing over many decades due to two important factors: the miniaturization of controllable systems, which raised the importance of a detailed accounting of dissipated energy in quantum devices, and fruitful thoughts on an old problem known as Maxwell's demon. The latter goes back to 1867, when Maxwell first suggested a challenge to the second law of thermodynamics, whereby an intelligent mechanism (i.e., a demon) is used to reduce the entropy of a gas without performing work on the system, seemingly violating the second law \cite{jaynes1957information, capek2005challenges, maruyama2009colloquium}. However, it was Szilard in 1929 who first made the connection between this entity and information theory. Assuming the correctness of the second law, Szilard showed that work extraction from a single heat reservoir at temperature $T$, using Maxwell's demon, requires the production of an amount of entropy equal to $k_{B}\ln 2$ per gas particle. In 1961, Landauer showed that the origin of this entropy production is due to the resetting of the memory of the demon, which requires dumping an amount of heat equal to $k_{B}T\ln 2$ into the thermal reservoir of temperature $T$. This has become the celebrated ``Landauer erasure principle,'' which has been extended to various scenarios in quantum information and quantum computing \cite{GPW05,bedingham2016thermodynamic, abdelkhalek2016fundamental,BBMW18,BBMW18pra, deffner2021energetic}.

One interesting area in which Landauer's principle has been applied is quantum error correction (QEC). The latter is of importance in quantum computation \cite{shor1995scheme,nielsen2002quantum,lidar_brun_2013}, where the state of the computational system (usually a collection of qubits) is corrected by first encoding it into a larger physical system (by adding redundant parts to the computational system) and then effectively reversing the action of the environmental noise using a recovery channel \cite{BK02,Ng_2010,tyson10,Junge2018, kwon2021reversing}. Hence, entropic and thermodynamic analyses have been conducted \cite{vedral2000landauer, korepin2002thermodynamic, cafaro2014entropic} to ensure that the effective reverse dynamics in QEC does not violate the second law of thermodynamics. However, the question of whether the laws of thermodynamics can add upon the existing QEC literature has been an open problem.

One major area where the field of thermodynamics has been contributing to QEC is in the theory of quantum measurements. In the latter, the ``information'' derived from a measurement can be quantified by the average entropy reduction of the measured system, which is also known as the Groeneworld information gain \cite{groenewold1971problem, ozawa-groen}. Therefore, this quantity can serve as a measure of information gain during the error-detection stage of QEC. The Groenewold information gain also has a thermodynamic meaning in terms of dissipated measurement heat \cite{jacobs2009second, buscemi2016approximate}, but this has hitherto been contingent on it taking non-negative values.

A richer and more direct link between QTD and QEC has been known since Ref.~\cite{nielsen1998information}, where the latter can be understood as a heat engine with a feedback controller \cite{landi2020thermodynamic} (see Figure~\ref{fig:1}). This link suggests that the laws of thermodynamics can be extended to the QEC setting. However, it is known that the second law for engines with feedback controllers varies slightly from that of typical Carnot engines. For instance, a critical thermodynamic feature of the former is that they can accomplish efficiencies higher than the Carnot efficiency  \cite{sagawa2008second}. Namely, Carnot's formulation of the second law is modified for such engines, since they have the advantage of measurement and feedback that typical Carnot engines do not. Of course, such a modified statement about possible super-Carnot efficiencies has to be stated very carefully in order to avoid some of the historical inconsistencies mentioned above; this has been done with some success in previous literature \cite{landi2020thermodynamic, sagawa2008second}. However, one has to be certain to erase the memory of Maxwell's demon (a classical register), as well as systematically consider the measurement heat during the feedback process \cite{abdelkhalek2016fundamental}. 

In this article, we expand the link between QTD and QEC by deriving thermodynamic constraints on the tasks of information gain from a measurement and QEC. These results can be summarized as follows:
\begin{enumerate}
    \item \textbf{Measurement heat is bounded from above by Groenewold information gain:} We show that, if the measurement apparatus is initialized in a thermal state (e.g., see \cite{allahverdyan2013understanding, allahverdyan2017sub}), then the heat dissipated into it (from the measured system) during a general quantum measurement is bounded from above by the Groenewold information gain. This extends the existing thermodynamic interpretation of the Groenewold information gain to cases in which it is negative~\cite{ozawa-groen}. Previously, only efficient quantum measurements have been discussed, for which the Groenewold information gain is always non-negative and the usual interpretation applies \cite{sagawa2008second, jacobs2009second, buscemi2008global}.
    
    \item \textbf{The second law of thermodynamics constrains quantum error correction:} We derive the second-law inequality in the context of QEC, which is accomplished in a more general setting than previously achieved in the literature \cite{sagawa2008second, landi2020thermodynamic}. More precisely, we allow for an arbitrary initial state of the system (which need not be a collection of qubits), arbitrary thermal noise (without an i.i.d.~assumption), and arbitrary generalized measurements at the error-detection stage. Achieving such a generalization is important to capture the full set of limitations that the laws of thermodynamics impose on QEC in a general physical setting.
    
    \item \textbf{Efficiency-fidelity trade-off in QEC engines:} We show that, under a set of physically motivated assumptions, the above thermodynamic limitations on QEC can be expressed in terms of a triple trade-off relation between the maximum error-correction fidelity of a QEC engine (described by Figure~\ref{fig:1}), the thermodynamic efficiency, and the efficacy of the quantum measurement at the error-detection stage of QEC. This can be broken up into two trade-off relations that can be summarized as follows: \textit{(i)} the larger the thermodynamic efficiency beyond the Carnot limit, the smaller the maximum achievable error-correcting fidelity of a QEC engine and \textit{(ii)} the larger the efficacy of the quantum measurement conducted at the error-detection stage, the smaller the thermodynamic efficiency below the Carnot limit. As an additional result, we clarify the conditions under which one can arrive at a super-Carnot efficiency regime, improving upon previous literature \cite{sagawa2008second}.
\end{enumerate}

\begin{figure*}
    \centering
    \includegraphics[width=\linewidth]{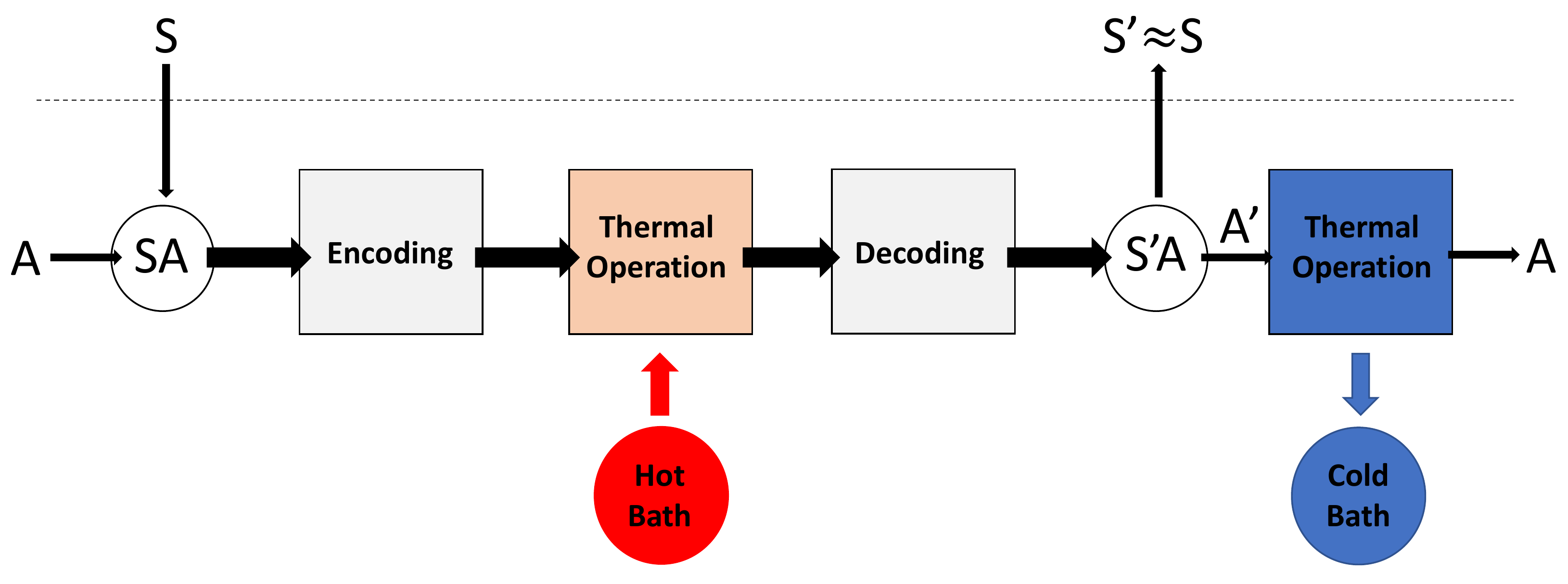}
    \caption{A diagrammatic representation of QEC as a heat engine with feedback controller, with time flow going from left to right). 
    From the left, we have the encoding of the system $S$ using the ancillary system $A$ for redundancy. Next, the encoded system $SA$ interacts with a hot bath of temperature $T_{h}$ (representing the noisy channel), followed by a decoding stage, where the errors in $SA$ are detected and an attempt to correct them is made via a unitary feedback process. Finally, the ancillary system is returned (recycled) to its original state with the help of a cold bath of temperature $T_{c}$, in preparation for a future use.}
    \label{fig:1}
\end{figure*}

\section{Preliminaries}
\label{sec:preliminaries}

In what follows, we review various information-theoretic and thermodynamic quantities that will be used in the rest of the article.

Let $\mathcal{H}$ denote a Hilbert space, and let $\mathcal{L(H)}$ be the set of bounded linear operators acting on~$\mathcal{H}$. Let $\mathcal{L_{+}(H)}$ denote the set of positive semi-definite operators acting on~$\mathcal{H}$. The state of a physical system is described by a density matrix $\rho \in \mathcal{D}(\mathcal{H})$, where $\mathcal{D}(\mathcal{H})$ is the subset of linear operators $\mathcal{L}(\mathcal{H})$ that are positive semi-definite and have unit trace. We denote by $H(A)_{\rho}\coloneqq -\operatorname{Tr}[\rho^{A} \ln \rho^{A}]$ the von Neumann entropy of the density matrix $\rho^{A} \in \mathcal{D}(\mathcal{H}^{A})$ of a system $A$. For a bipartite system $AB$ described by a density matrix $\sigma^{AB} \in \mathcal{D}(\mathcal{H}^{A}\otimes \mathcal{H}^{B})$, the quantum (von Neumann) conditional entropy and mutual information are denoted by $H(A|B)_{\sigma}\coloneqq H(AB)_{\sigma} - H(B)_{\sigma}$ and $I(A\!:\!B)_{\sigma} \coloneqq H(A)_{\sigma} - H(A|B)_{\sigma}$, respectively. Given two operators $M,N \in \mathcal{L_{+}(H)}$, we denote their quantum relative entropy by $D(M \Vert N)$, which is non-symmetric and defined as
\begin{equation}
    D(M\Vert N) \coloneqq \mathrm{Tr}\left[M\ln M\right]-\mathrm{Tr}\left[M\ln N\right] \;,
\end{equation}
if supp($M$) $\subseteq$ supp($N$), and $+\infty$ otherwise, where supp($\cdot$) denotes the support (i.e., the orthogonal complement of the kernel) of an operator. 

A linear map from $\mathcal{L}(\mathcal{H}^{A})$ to  $\mathcal{L}(\mathcal{H}^{B})$ is denoted by $\mathcal{N}^{A\rightarrow B}:\mathcal{L}(\mathcal{H}^{A})\rightarrow \mathcal{L}(\mathcal{H}^{B})$. We say that a linear map is positive if  $\mathcal{N}^{A\rightarrow B}(M^{A}) \in \mathcal{L}_{+}(\mathcal{H}^{B})$ for all $M^{A} \in \mathcal{L}_{+}(\mathcal{H}^{A})$, and trace preserving (TP) if $\mathrm{Tr}[\mathcal{N}^{A\rightarrow B}(M^{A})]=\mathrm{Tr}[M^{A}]$ for all $M^{A} \in \mathcal{L}(\mathcal{H}^{A})$. Similarly, we say that $\mathcal{N}^{A\rightarrow B}$ is trace non-increasing (non-decreasing) if $\mathrm{Tr}[\mathcal{N}^{A\rightarrow B}(M^{A})] \leq \left(\ge\right)\, \mathrm{Tr}[M^{A}]$ for all $M^{A} \in \mathcal{L}_{+}(\mathcal{H}^{A})$. A linear map $\mathcal{N}^{A\rightarrow B}$ is unital if it maps the unit operator in $\mathcal{L}(\mathcal{H}^{A})$ to the unit operator in $\mathcal{L}(\mathcal{H}^{B})$, i.e.,  $\mathcal{N}^{A\rightarrow B}(I^{A})=I^{B}$. Similarly, a linear map $\mathcal{N}^{A\rightarrow B}$ is subunital (superunital) if $\mathcal{N}^{A\rightarrow B}(I^{A}) \leq (\ge) I^{B}$. We define the adjoint map $\left( \mathcal{N}^{A\rightarrow B} \right)^{\dagger}$ of any linear map $\mathcal{N}^{A\rightarrow B}$ through the relation $\mathrm{Tr}[M\mathcal{N}^{A\rightarrow B}(N)]=\mathrm{Tr}[(\mathcal{N}^{A\rightarrow B})^{\dagger}(M)N]$ for all $N \in \mathcal{L}(\mathcal{H}^{A})$ and $M \in \mathcal{L}(\mathcal{H}^{B})$. A linear map is unital if and only if the corresponding adjoint map is TP. Similarly, a linear map is subunital (superunital) if and only if the corresponding adjoint map is trace non-increasing (non-decreasing). A positive linear map $\mathcal{N}^{A\rightarrow B}$ is called completely positive (CP) if for every Hilbert space $\mathcal{H}^{R}$, the map $\mathcal{I}^{R}\otimes \mathcal{N}^{A\rightarrow B}$ is positive, where $\mathcal{I}^{R}$ is the identity map acting on $\mathcal{L}(\mathcal{H}^{R})$. It is well known that CP maps have a Kraus decomposition.

For any CP linear map $\mathcal{N}^{A\rightarrow B}$ acting on a density matrix $\rho^{A} \in \mathcal{D}(\mathcal{H}^{A})$, we define the \textit{efficacy} $\mathscr{E}$ \cite{albash2013fluctuation, goold2015nonequilibrium, buscemi2020thermodynamic} of the linear map $\mathcal{N}^{A\rightarrow B}$ with respect to $\rho^{A}$ as 
\begin{equation}
    \mathscr{E}(\mathcal{N}^{A \rightarrow B}; \rho^{A}) \coloneqq \mathrm{Tr}\left[(\mathcal{N}^{A\rightarrow B})^{\dagger}\circ \mathcal{N}^{A\rightarrow B} (\rho^{A})\right]\;, \label{eqn:efficacy}
\end{equation}
which is a measure of how ``reversible'' $\mathcal{N}^{A\rightarrow B}$ is with respect to  $\rho^{A}$: indeed, it can be shown that the efficacy is related to the entropy change as follows \cite{buscemi2016approximate,buscemi2020thermodynamic} 
\begin{align}
H(B)_{\mathcal{N}(\rho)}-H(A)_\rho&\ge D(\rho^{A}\Vert (\mathcal{N}^{A\rightarrow B})^{\dagger}\circ \mathcal{N}^{A\rightarrow B} (\rho^{A}))\nonumber\\
&=D(\rho^{A}\Vert \tilde{\rho}^{A})-\ln \mathscr{E}(\mathcal{N}^{A \rightarrow B}; \rho^{A})\nonumber\\
&\ge -\ln \mathscr{E}(\mathcal{N}^{A \rightarrow B}; \rho^{A})\;, \label{eqn:ent-to-efficacy}
\end{align}
where
\begin{equation}
    \tilde{\rho}^{A} \coloneqq \frac{(\mathcal{N}^{A\rightarrow B})^{\dagger}\circ \mathcal{N}^{A\rightarrow B} (\rho^{A})}{\mathrm{Tr}\left[ (\mathcal{N}^{A\rightarrow B})^{\dagger}\circ \mathcal{N}^{A\rightarrow B} (\rho^{A}) \right]} \in \mathcal{D}(\mathcal{H}^{A})\;,
\end{equation}
is a valid density matrix. It is easy to see that, for efficacy to be strictly smaller (larger) than one for a fixed $\rho^{A}$, it is sufficient (but not necessary) for $\mathcal{N}^{A\rightarrow B}$ to be strictly subunital (superunital). Furthermore, due to Eq.~\eqref{eqn:ent-to-efficacy}, a subunital $\mathcal{N}^{A\rightarrow B}$ implies a positive entropy change, but a superunital $\mathcal{N}^{A\rightarrow B}$ does not necessarily imply negative entropy change. See also Ref.~\cite{das2018fundamental} for a measure called diamond distance of non-unitarity, which is similar in spirit to the efficacy. 

In our article, we need a formalism that characterises general quantum measurements, which can be thought of as a linear mapping between the pre- and post-measurement states, along with some measurement statistics. For that, we define a \textit{quantum instrument} \cite{ozawa1984quantum,wilde2011classical, hayashi2006quantum} as a set $\left\{ \mathcal{N}^{A}_{x}\right\}_{x\in \mathcal{X}}$ of completely positive trace-non-increasing linear maps such that $\mathcal{N}^{A}=\sum_{x\in \mathcal{X}}\mathcal{N}^{A}_{x}$ is CPTP. Quantum instruments describe generalized quantum measurements with measurement outcomes $x\in \mathcal{X}$. The probability of each measurement outcome $x$ is computed by $p_{X}(x)\coloneqq \mathrm{Tr}[\mathcal{N}^{A}_{x}(\rho^{A})]$ and the corresponding post-measurement state by
\begin{equation}
    \rho^{A}\rightarrow \theta^{A}_{x} \coloneqq\frac{ \mathcal{N}^{A}_{x}(\rho^{A})}{p_{X}(x)}\;.
\end{equation}
It is convenient to record the measurement outcome in an auxiliary classical register $X$ that is initially in a pure state $|0\rangle\!\langle 0|^{X}$. In this way, we can write the expected post-measurement state as a quantum-classical state:
\begin{equation}
    \theta^{AX} \equiv \mathcal{N}^{A\rightarrow AX}(\rho^{A})=\sum_{x\in \mathcal{X}}\mathcal{N}^{A}_{x}(\rho^{A})\otimes |x\rangle\!\langle x|^{X}\;, \label{eqn:quantum_instrument}
\end{equation}
where $\left\{|x\rangle^{X}\right\}_{x\in \mathcal{X}}$ is an orthogonal set of ``pointer states'' of the classical register.

Another important concept in this work is the notion of information gain of generalized quantum measurements. One such attempt in quantifying information gain from a measurement is given by the Groenewold information gain \cite{groenewold1971problem, ozawa-groen}, which is defined as the expected entropy reduction of the system due to a general measurement described by the quantum instrument $\left\{ \mathcal{N}^{A}_{x}\right\}_{x \in \mathcal{X}}$:
\begin{align}
    I_{G}(\{ \mathcal{N}^{A}_{x}\}_{x}; \rho^{A})&\coloneqq H(A)_{\rho}-\sum_{x\in \mathcal{X}}p_{X}(x)H(A)_{\theta_{x}}\\
    &=H(A)_{\rho}-H(A|X)_{\theta}\;.
    \label{eqn:groenewold_def}
\end{align}
When no confusion arises, in what follows we will use the shorthand notation $I_G$ instead of $I_{G}(\{ \mathcal{N}^{A}_{x}\}_{x}; \rho^{A})$. Although non-negative for projective measurements, the Groenewold information gain can, in general, become negative~\cite{ozawa-groen}. For example, this is easy to verify for an instrument that outputs the maximally mixed state $\theta^{A}_{x}=I^{A}/d^{A}$ for every measurement outcome $x \in \mathcal{X}$, independently of the input state. The negativity of the Groenewold information gain has arguably hindered its operational interpretation, and indeed previous researchers have proposed alternative measures of information gain that are always non-negative, possess an operational interpretation, and reduce to Groenewold information gain for projective measurements. For example, Ref.~\cite{buscemi2008global} defined the ``quantum information gain'' as the mutual information between the pre-measurement purifying reference system $R$ of $A$ and the classical register~$X$. It can be shown \cite{buscemi2008global} that the Groenewold information gain $I_{G}$ coincides with the quantum information gain (hence, $I_{G} \ge 0$) for efficient measurements (i.e., when $\mathcal{N}_{x}(\cdot)=N_{x}(\cdot)N^{\dagger}_{x}$, which maps pure states to pure states). Additionally, the quantum information gain has a compelling operational interpretation in terms of an information-processing task called measurement compression \cite{Winter01a} (see also \cite{wilde2012information,BRW14,LWD16,AJW19,AHP19}).

As we will end up discarding the classical register $X$ after using the outcome to perform QEC, it is important to consider the entropy reduction due to this procedure, which is given by
\begin{equation}
    H(AX)_{\theta}-H(A)_{\theta}=H(X|A)_{\theta} \geq 0\;, \label{eqn:cond_ent}
\end{equation}
where the non-negativity follows from the fact that quantum conditional entropy $H(X|A)_{\theta} =H(AX)_{\theta}-H(A)_{\theta} \ge 0$ is non-negative for every separable state. The state in this case can be written as $\theta^{AX}=\sum_{x \in \mathcal{X}}p_{X}(x)\theta^{A}_{x}\otimes |x\rangle \! \langle x|^{X}$ and is thus separable. It is known that the quantity $H(X|A)$ has an information-theoretic meaning as the compression rate in the task of classical data compression with quantum side information \cite{DW03}.
Thermodynamically, it is easy to see that the discarding process of $X$ from the joint system $AX$, characterized by the entropy change $H(AX)_{\theta}-H(A)_{\theta}$, is equivalent to the erasure of the classical register state $\theta^{X}\rightarrow |0\rangle\!\langle 0|^{X}$ (or in other words, erasing the memory of Maxwell's demon). See Ref.~\cite{Rio2011} for a study of thermodynamic erasure cost in the presence of a quantum memory, where quantum conditional entropy was shown to be the optimal erasure cost.

In this article, we also need to quantify the success of the QEC procedure: this can be accomplished by using a quantity called \textit{entanglement fidelity} $F_{e}$~\cite{schumacher1996quantum, schumacher1996sending}. If the initial state $\rho_{i}^{A}$ of a system $A$ is purified via a reference system $R$, i.e., $\rho_{i}^{A}\rightarrow \psi^{RA}$ with $\mathrm{Tr}_{R}[\psi^{RA}]=\rho^{A}_{i}$, then the entanglement fidelity is defined as
\begin{equation}
    F_{e}=\langle \psi^{RA}|\rho^{RA}_{f}|\psi^{RA}\rangle\;, \label{eqn:fidelity}
\end{equation}
where $\rho^{RA}_{f}=(\operatorname{id}\otimes(\mathcal{R}\circ \mathcal{E}))(\psi^{RA})$ is the state of $RA$ after the application of the noisy channel $\mathcal{E}$ and then the correcting channel $\mathcal{R}$. Physically, $F_{e}$ is equal to the probability that the state $\rho^{RA}_{f}$ passes a test for being the initial state $\psi^{RA}$ (in particular, the test can be conducted by performing the  measurement $\left\{ |\psi\rangle\!\langle \psi|^{RA}, I^{RA}-|\psi\rangle\!\langle \psi|^{RA} \right\}$) \cite{wilde2011classical}. Relevant properties of entanglement fidelity include being independent of the purification $R$ and being a lower bound on the input-output fidelity $F(\rho^{A}_{f}, \rho^{A}_{i})$ \cite{schumacher1996sending}, where the latter does not reflect the preservation (or the lack thereof) of entanglement that might be shared between $A$ and some other system (e.g., $R$ above).

It is important to note that when talking about QEC, it is not generally possible to correct for all possible preparations of the quantum state of $A$, but rather for a subset $\rho^{A} \in \mathcal{C} \subseteq  \mathcal{D}(\mathcal{H}^{A})$ of such states, where $\mathcal{C}$ is known as the codespace. Given a fixed noisy channel~$\mathcal{E}$, a necessary and sufficient condition for the existence of a codespace and a CPTP map $\mathcal{R}$ on that space that accomplishes perfect QEC ($F_{e}=1$) is given by the Knill--Laflamme theorem \cite{knill1997theory, ogawa2005perfect}. It is also worth noting that the state of the reference system does not change as a result of any local CPTP operation $\mathcal{I}^{R}\otimes \mathcal{M}^{A}$ on $A$. This can be easily seen by checking that $\mathrm{Tr}_{A}[(\mathcal{I}^{R}\otimes \mathcal{M}^{A})(\psi^{RA})]=\mathrm{Tr}_{A}[\psi^{RA}]$.

A major physical concept that plays a fundamental role in both fields of QTD and QEC is that of the irreversibility of the system-environment evolution $\rho^{A}_{i}\rightarrow \rho^{A}_{f}=\mathrm{Tr}_{B}[U(\rho^{A}_{i}\otimes \sigma^{B}_{i})U^{\dagger}]$, initiated from a product state. In QTD, this is typically quantified by $\Delta H(B)$, which is the net entropy dissipated into the environment $B$ \cite{bedingham2016thermodynamic}. In QEC, this quantity has been given the name \textit{entropy exchange} \cite{schumacher1996quantum, schumacher1996sending}. Here, the initial state of the environment is taken to be pure (without loss of generality \cite{mynote1}), and the initial state of $A$ is purified using a reference system $R$. Then, the entropy exchange $S_{e}$ is identified with the final state of $RA$, namely $S_{e}\coloneqq H(RA)_{\rho_{f}}$. It can be shown that this quantity is independent of the purification and, more importantly, that it serves as an upper bound via $S_{e} \ge |\Delta H(B)|$ \cite{schumacher1996sending}, where the inequality is saturated for a pure initial state of $B$. The latter property assigns a physical meaning to entropy exchange as the amount of entropy dissipated into an environment of an initially pure state.

Finally, in order to show what constraints thermodynamics can add to QEC, we have to define the concepts of thermodynamic work and heat. A more detailed discussion is contained in Appendix~\ref{sub:work_heat}. Given two systems $A$ and $B$, we can functionally separate a system of interest and a heat bath in the following way \cite{balian2007microphysics}: $(i)$ we assume that the heat bath $B$ is always prepared in an initially thermal state $\tau^{B}$ of temperature $T$, and $(ii)$ we assume that the system of interest $A$ is controllable (in contrast to the heat bath), via an external control (work) parameter~$\lambda_{t}$ with predetermined trajectory $\left\{ \lambda_{t} \right\}_{t \ge 0}$. The latter is reflected in the total Hamiltonian dynamics of $AB$ by the fact that the local Hamiltonian of $A$ (denoted by $\mathscr{H}^{A}$) depends implicitly on time via its explicit dependence on the work parameter $\lambda_{t}$, i.e., $\mathscr{H}^{A}_{\lambda_{t}}$. After making these two important distinctions between $A$ and $B$, we define the internal energy $\mathscr{U}^{A}_{t}$ of $A$ at time $t$, in the weak coupling limit, as \begin{equation}
    \mathscr{U}^{A}_{t}=\mathrm{Tr}[\rho^{A}_{t} \mathscr{H}^{A}_{\lambda_{t}} ]\;.
\end{equation}
From here, thermodynamic work and heat can be identified with two different contributions in the internal energy change of $A$ \cite{balian2007microphysics} via
\begin{equation}
    \partial_{t}\mathscr{U}^{A}_{t}
    =
    \mathrm{Tr}[(\partial_{t}\rho^{A}_{t}) \mathscr{H}^{A}_{\lambda_{t}}]+\mathrm{Tr}[ \rho^{A}_{t} (\partial_{t}\mathscr{H}^{A}_{\lambda_{t}}) ]\;, \label{eqn:work_heat}
\end{equation}
where the time integral of the first term defines heat $Q^{A}_{t}$ and that of the second defines work $W^{A}_{t}$. To validate the intuition behind these definitions, note that if the system $A$ is isolated (i.e., $B$ is not present and hence we expect no heat to be dissipated), then $\rho^{A}_{t}$ satisfies the Liouville-von Neumann equation, and the first term in Eq.~\eqref{eqn:work_heat} becomes zero due to the cyclicity of trace  $\mathrm{Tr}[CD]=\mathrm{Tr}[DC]$. On the other hand, if the external control parameter is not varied (hence we expect no work to be performed by or on $A$), then the second term in Eq.~\eqref{eqn:work_heat} becomes trivially equal to zero.

\section{Thermodynamics of General Quantum Instruments}

It is known \cite{ozawa1984quantum} that each CP map $\mathcal{N}^{A}_{x}$ of the quantum instrument $\{ \mathcal{N}^{A}_{x} \}_{x\in \mathcal{X}}$ can be written in terms of a common unitary interaction between the system $A$ and a measurement apparatus (probe) $M$, followed by a projective measurement on $M$ (also known as an indirect measurement of system $A$ \cite{ozawa1984quantum,breuer2002theory,hayashi2006quantum}), as follows
\begin{equation}
     \mathcal{N}^{A}_{x}(\rho^{A})=\mathrm{Tr}_{M} \!\left[ U^{AM}(\rho^{A}\otimes \sigma^{M})(U^{AM})^{\dagger}P^{M}_{x}\right]\;, \label{eqn:indirect_meas}
\end{equation}
where $\sigma^{M}\in \mathcal{D}(\mathcal{H}^{M})$ is the quantum state of the measurement apparatus, $\mathcal{H}^{M}$ is the Hilbert space of $M$, $U^{AM}$ is a global unitary acting on the joint system $AM$, and $\{ P^{M}_{x} \}_{x \in \mathcal{X}}$ is a complete set of projectors. We call the tuple $(\sigma^{M}, U^{AM}, \{ P^{M}_{x} \}_{x \in \mathcal{X}})$ an indirect measurement model~\cite{ozawa1984quantum} (or implementation) of the quantum instrument $\left\{ \mathcal{N}^{A}_{x}\right\}_{x\in \mathcal{X}}$. To a given quantum instrument, there is an infinite number of implementations, all yielding the same measurement statistics $\{ p_{X}(x) \}_{x \in \mathcal{X}}$ and the post-measurement states $\{\theta^{A}_{x}\}_{x\in \mathcal{X}}$ of the system $A$, but involving very different physical models in principle.

For the purposes of this article, it is important to note that the quantum instrument formalism can also be thought of as a unitary interaction between the system $A$, a measurement apparatus $M$, and a classical register $X$, as
\begin{equation}
    \theta^{AMX}=U^{AMX}(\rho^{A}\otimes \sigma^{M} \otimes |0\rangle \! \langle 0|^{X})(U^{AMX})^{\dagger}\;, \label{eqn:unitary_meas}
\end{equation}
with $U^{AMX}=V^{MX}U^{AM}$, where $U^{AM}$ is the unitary operator in~\eqref{eqn:indirect_meas} and $V^{MX}$ is another unitary operator performing a controlled transformation on the register system~\cite{abdelkhalek2016fundamental}
\begin{equation}
    V^{MX}=\sum_{x \in \mathcal{X}}P^{M}_{x}\otimes V^{X}_{x}\;,
\end{equation}
where $V^{X}_{x}|0\rangle^{X}=|x\rangle^{X}$. We can easily check that this unitary yields $\theta^{AX}=\mathrm{Tr}_{M}[\theta^{AMX}]$ with $\theta^{AX}$ given by Eq.~\eqref{eqn:quantum_instrument}.

\subsection{Measurement Heat}
\label{sub:meas_heat}

Heat is one of the central quantities of thermodynamics \cite{balian2007microphysics}, though its rigorous definition (along with the definition of internal energy) becomes problematic in certain regimes, such as in the presence of strong coupling~\cite{esposito2010entropy, kwon2019three} (please see Appendix~\ref{sub:work_heat}). However, in the weak-coupling regime, where contributions to the internal energy due to the interaction Hamiltonian can be neglected, the following definition is accepted: for a system $S$ of interest coupled with a thermal bath $B$, the heat $Q^{B\to S}_{t}$ absorbed by $S$ from $B$ or, equivalently, the heat $Q^{S\to B}_t$ dissipated from $S$ into $B$ is given by
\begin{equation}
    -Q^{B\to S}_{t}\equiv Q^{S\to B}_t= \Delta \langle \mathscr{H}^{B}\rangle \coloneqq \langle \mathscr{H}^{B}\rangle_{t}-\langle \mathscr{H}^{B}\rangle_{0}\;, \label{eqn:heat_def}
\end{equation}
where $\langle \mathscr{H}^{B}\rangle_{0}=\mathrm{Tr}[\tau^{B}\mathscr{H}^{B}]$ and $\langle \mathscr{H}^{B}\rangle_{t}=\mathrm{Tr}[\rho^{B}_{t}\mathscr{H}^{B}]$ denote the expectation of the bath Hamiltonian $\mathscr{H}^{B}$ at the initial and final times, respectively. In Appendix~\ref{sub:work_heat}, we review the definitions of work and heat in more detail. There, we recall that under a set of three assumptions: $(i)$ the interaction between $S$ and $B$ is unitary; $(ii)$ $B$ is initialized in a thermal state $\tau^{B}$ of inverse temperature $\beta=1/k_{B}T$; and $(iii)$ systems $S$ and $B$ are initially uncorrelated, one has~\cite{reeb2014improved}
\begin{equation}
    \beta Q_t^{S\to B}=-\Delta H(S)+ I(S:B)+D(\rho^{B}_{t}\Vert\tau^{B})\; . \label{eqn:landauer_eqn}
\end{equation}
There have been various efforts to characterize the heat absorbed or dissipated during a quantum measurement process, e.g., \cite{jacobs2009second, bedingham2016thermodynamic, abdelkhalek2016fundamental}. Here, we propose a definition of measurement heat of a quantum instrument $\{ \mathcal{N}^{A}_{x} \}_{x \in \mathcal{X}}$ in the special case where the measurement system (probe) $M$ is initially prepared in a thermal state, which is the natural choice when one is interested in the thermodynamics of quantum measurements, e.g., \cite{allahverdyan2013understanding, allahverdyan2017sub}. For that, recall that the quantum instrument formalism of a general quantum measurement can be described as a fully unitary interaction model between the system $A$ being measured, the measurement apparatus $M$, and a classical memory $X$, as described in Eq.~\eqref{eqn:unitary_meas}. Therefore, if the measurement apparatus $M$ is initially in a thermal state of temperature $T_{m}$, then the assumptions leading to Eq.~\eqref{eqn:landauer_eqn} apply, where the joint system $AX$ plays the role of $S$, while $M$ plays the role of the bath $B$. 
Accordingly, we have
\begin{align}
	 &\beta Q_{\operatorname{meas}} \notag\\
	& \coloneqq\beta Q^{M\to AX} \notag\\
	&=\Delta H(AX)- I(AX:M)_\theta-D(\theta^{M}\Vert\tau^{M})\notag\\
	&=H(AX)_\theta-H(A)_\rho- I(AX:M)_\theta-D(\theta^{M}\Vert\tau^{M})\notag\\
	&=H(X)_\theta-I_G- I(AX:M)_\theta-D(\theta^{M}\Vert\tau^{M})\;.
\end{align}
where $I_{G} \equiv I_{G}(\left\{\mathcal{N}^{A}_{x}\right\}_{x \in \mathcal{X}};\rho^{A})$ is the Groenewold information gain computed for the pre-measurement state $\rho^{A}$. From the above, we immediately see that a positive Groenewold information gain reduces the heat absorbed by $AX$ from the measurement apparatus, whereas a negative Groenewold information gain increases that heat.

A tighter characterization can be given if $\theta^M=\tau^M$, that is, if the final reduced state of the measurement apparatus remains thermal. This is a plausible assumption for a macroscopic measurement apparatus, whose internal state is not affected by the interaction with system and memory. In such a case, $D(\theta^{M}\Vert\tau^{M})=0$. We can then write
\begin{align}
	\beta Q_{\operatorname{meas}}&=\Delta H(AX)- I(AX:M)_\theta\\
	&=:\beta Q_{\operatorname{meas}}^{X}+\beta Q_{\operatorname{meas}}^{A|X}\;,
\end{align}
where we have used the chain rule of quantum mutual information~\cite{wilde2012information} and the definitions:
\begin{align}
    \beta Q_{\text{meas}}^{X}& \coloneqq \Delta H(X)-I(X:M)_{\theta}=H(X|M)_{\theta},\\
    \beta Q_{\text{meas}}^{A|X}& \coloneqq \Delta H(A|X)-I(A:M|X)_{\theta},
\end{align}
where $\Delta H(X)=H(X)_{\theta}$ because the classical memory $X$ is initialized from the pure state $|0\rangle \! \langle 0|^{X}$. Note that, on the one hand, the quantity $Q^{X}_{\text{meas}}$ is the amount of heat absorbed by the classical register $X$ in a general quantum measurement (involving $A$, $X$, and $M$), which provides an alternative physical meaning to the conditional quantum entropy $H(X|M)_{\theta}$ described after Eq.~\eqref{eqn:cond_ent}, for the special case of a quantum measurement. On the other hand, since $H(A|X)_{\rho}=H(A)_{\rho}$, we can also write $Q_{\text{meas}}^{A|X}$ as
\begin{align}
	\beta Q_{\operatorname{meas}}^{A|X}=-I_G-I(A:M|X)_\theta\le -I_G\;,
\end{align}
where the latter does not depend on $M$. We summarize the above discussion as follows:

\begin{theorem}
	Suppose that a quantum system $A$ undergoes a general quantum measurement described by the quantum instrument $\{ \mathcal{N}^{A}_{x} \}_{x \in \mathcal{X}}$. Further, suppose that the measurement is physically implemented according to an indirect measurement $(\sigma^{M}, U^{AM}, \{ P^{M}_{x} \}_{x \in \mathcal{X}})$ where $\sigma^{M}$ (the measurement apparatus' internal state) is thermal at temperature $T_{m}$. Suppose moreover that the reduced state of the measurement apparatus, after the measurement interaction, remains unchanged (it can become correlated with the other systems however). Then, the heat absorbed by the system $A$ from the measurement apparatus $M$ for each outcome $x\in\mathcal{X}$, averaged over all outcomes, is bounded from above by the negative of the Groenewold information gain:
	\begin{equation}
		Q_{\operatorname{meas}}^{A|X}\leq -k_{B}T_{m}I_{G}\;. \label{eqn:result_groenewold}
	\end{equation}
	Equivalently, the heat dissipated from the system into the measurement apparatus is bounded from below by the Groenewold information gain.
\end{theorem}

To our knowledge, this is the first attempt to provide a physical meaning to the Groenewold information gain even in the case when it is negative, which represents the vast majority (if not the totality) of all practically achievable quantum measurements. In fact, the Groenewold information gain is always non-negative only if the quantum measurement is \textit{quasicomplete}~\cite{ozawa-groen}, that is, when the post-measurement states $\mathcal{N}_x^A(\rho^A)$ are all pure whenever $\rho^A$ is pure: a condition which is arguably highly idealized.

\section{Thermodynamics of QEC Engines}
\label{QECengine}

In this section, we shall introduce the QEC engine and analyse its various stages of operation. The ultimate goal in this section is to derive a second law of thermodynamics for QEC engines that holds for an arbitrary initialization of the system, general (non-i.i.d.) thermal noise, and arbitrary quantum measurements, thereby generalizing results such as Ref.~\cite{sagawa2008second}. 
\subsection{Engine Components}
We consider an arbitrary thermal operation $\mathcal{T}$ as the noisy channel that we would like to correct, namely
\begin{equation}
    \mathcal{T}(\rho)=\mathrm{Tr}_{B}\left[U(\rho \otimes\tau^{B})U^{\dagger}\right]\;, \label{eqn:thermal_op}
\end{equation}
where $\tau^{B}$ is a thermal state of the bath and $U$ is a unitary operator that commutes with the total Hamiltonian \cite{Janzing2000}. We shall refer to the environment in this representation as the ``hot'' bath $B_{h}$ with temperature $T_{h}$ (inverse temperature $\beta_{h}$). We denote the system of interest and the purifying reference system by $S$ and $R$, respectively. We also denote by $A$ the ancillary system that is used to encode the system's state.

In order to make a connection between QEC and heat engines with feedback, we need to introduce a cold bath $B_{c}$ with temperature $T_{c}$ (inverse temperature $\beta_{c}$). The purpose of the cold bath is to recycle the ancillary system after the QEC procedure, in preparation for future use. This raises the question of what should the initial preparation of $A$ be for this construction to work. In the literature, the ancillary system is usually prepared in an initially pure state (often uncorrelated qubits), but this might not be possible for practical systems subject to thermal noise. Furthermore, resetting $A$ back to a pure state \cite{landi2020thermodynamic} might imply that the cold bath has to operate at zero temperature $T_{c}=0$, which is not only unphysical, but also takes away from the generality of possible thermodynamic statements that can be made for arbitrary $T_{c}$. Hence, to maintain full generality for arbitrary ancillary systems $A$ and arbitrary temperatures $T_{c}$ of the cold bath, we assume that $A$ is prepared in an initially thermal state of the same temperature $T_{c}$. By denoting the energy basis of $A$ by $\left\{ |\varepsilon_{x}\rangle^{A} \right\}_x$, the initial state of the ancillary system is therefore given by
\begin{equation}
    \tau^{A}_{c}=\sum_{x}p_{X}(x)|\varepsilon_{x}\rangle\!\langle \varepsilon_{x}|^A \;,
\end{equation} 
with
\begin{equation}
  p_{X}(x)\coloneqq e^{-\beta_{c}\varepsilon_{x}}/Z \hspace{0.5cm} \text{where} \hspace{0.5cm} Z=\sum_{x}e^{-\beta_{c}\varepsilon_{x}}\;.
\end{equation}
In our article, we accomplish the recycling back to the thermal state $\tau^{A}_{c}$ by conducting a SWAP operation between $A$ and one of the uncorrelated subsystems of $B_{c}$ that possesses the same dimensions and local Hamiltonian as $A$.

After initializing $A$ in $\tau^{A}_{c}$, we project it onto one of its energy eigenstates, in preparation for encoding. The outcome of the measurement is recorded using a classical register $X$, for the purpose of using it in the encoding and decoding stages of QEC.
The last component in the QEC engine is the classical register $Y$ that records the result of the error-syndrome measurement (conducted by a measurement apparatus $M$) and feeds it into the decoding channel, hence playing the same role as a feedback controller in a heat engine. 
For brevity, we shall use $E$ for the combined engine components $R$, $S$, $A$, $B_{h}$, and $B_{c}$. To summarize, the QEC engine is comprised of the systems $E$, $X$, and $Y$, all of which are depicted in Figure~\ref{fig:2}.

\begin{figure*}
    \centerline{
    \begin{quantikz}[
    column sep=0.38cm, row sep=0.2cm]
    \lstick[wires=5]{$\sigma^{E}_{i}$}\hspace{0.6cm} & \lstick[wires=2]{$\psi^{RS}$} & \qw\slice{$\sigma^{EXY}_{i}$} & \qw\slice{$\sigma^{EXY}_{0}$} & \qw & \qw & \qw\slice{$\sigma^{EXY}_{\text{enc}}$}  & \qw & \qw & \qw\slice{$\sigma^{EXY}_{\text{1}}$} & \qw & \qw & \qw\slice{$\sigma^{EXY}_{\text{2}}$} & \qw & \qw & \qw\slice{$\sigma^{EXY}_{\text{3}}$} & \qw & \qw & \qw\slice{$\sigma^{EXY}_{f}$}  & \qw & \qw & \rstick[wires=2]{$\sigma^{RS}_{f}$} \\
    &  & \qw & \qw & \qw & \gate[wires=2]{U_{\text{enc}, x}} & \qw & \qw & \gate[wires=3]{U_{h}} & \qw & \qw & \gate[wires=5, nwires={3, 4}]{{\mathcal{N}_{y|x}}} & \qw & \qw & \gate[wires=2][2cm]{U_{\text{dec}, xy}} & \qw & \qw & \qw & \qw & \qw & \qw & \\
    & \lstick{$\tau^{A}_{c}$} & \qw & \ctrl{4} & \qw & & \qw & \qw & & \qw & \qw & & \qw & \qw & & \qw & \qw & \gate[wires=3, nwires={2}]{U_{c}} & \qw & \qw & \qw & \rstick{$\tau^{A}_{c}$}\\
    & \lstick{$\tau^{B_{h}}_{h}$} & \qw & \qw & \qw & \qw & \qw & \qw & & \qw & \qw & \linethrough &  & \qw & \qw & \qw & \qw & \linethrough &  & \qw & \qw \\
    & \lstick{$\tau^{B_{c}}_{c}$} & \qw & \qw & \qw & \qw & \qw & \qw & \qw & \qw & \qw & \linethrough &  & \qw & \qw & \qw & \qw & & \qw & \qw & \qw \\
    & \lstick{$|0\rangle \! \langle 0|^{Y}$} & \qw & \qw & \qw & \qw & \qw & \qw & \qw & \qw & \qw & \qw & \qw & \push{y} & \ctrl{-3} & \qw & \qw & \qw & \qw & \qw & \qw & \rstick[wires=2]{{discard}}\\
    & \lstick{$|0\rangle \! \langle 0|^{X}$} & \qw & \gate[nwires={1}]{V_{x}}\qw & \push{x} & \ctrl{-4} & \qw & \qw & \qw & \qw & \push{x} & \qw \vcw{-1} & \qw & \push{x} & \ctrl{-4}\qw & \qw & \qw & \qw & \qw & \qw & \qw & 
    \end{quantikz}
    }
    \caption{Circuit description of the QEC engine. Lines locally disconnected from channels imply that the global channel (at that time step) has the effect of a local identity operation on the corresponding systems. We emphasise that the internal Hamiltonian of the reference system is assumed to be trivially zero at the initial and final times, namely $\mathscr{H}^{R}_{i}=\mathscr{H}^{R}_{f} \equiv 0$. All other (local and interaction) Hamiltonians are assumed to be cyclic, i.e., $\mathscr{H}_{i}=\mathscr{H}_{f}$. The only non-unitary steps in our engine are  the error-syndrome measurement (characterised by the quantum instrument $\{ \mathcal{N}^{SA}_{y|x}\}_{y}$ for a given $x$) and the discarding (memory erasure) of the classical registers $X$ and $Y$.}
    \label{fig:2}
\end{figure*}
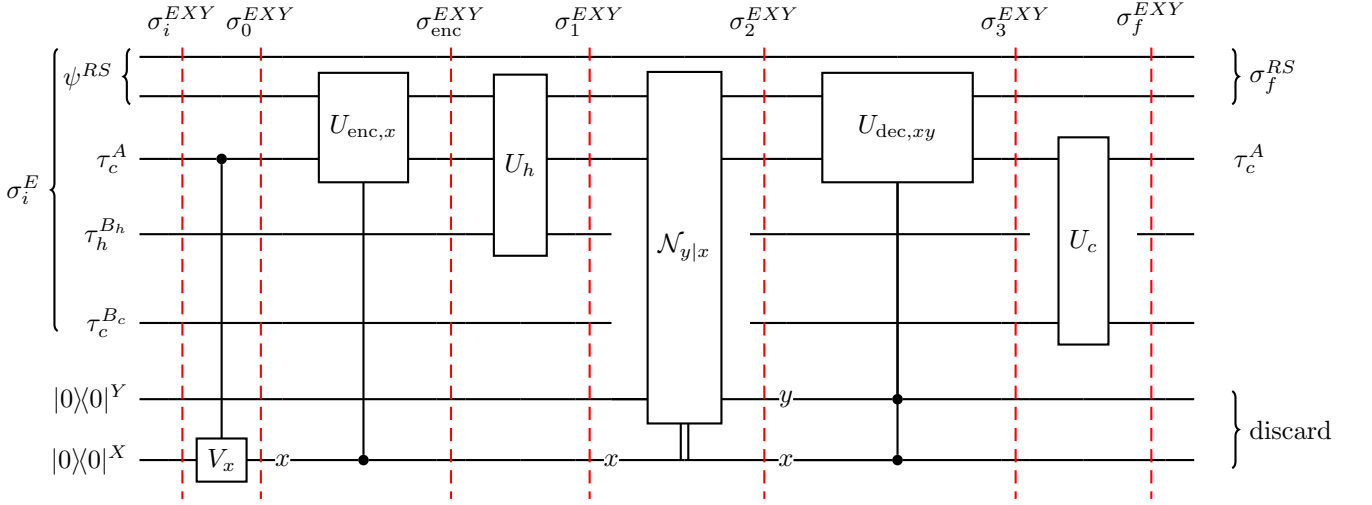

\subsection{Cycle Description}
\label{sub:stages}

We now describe in detail the operation of the QEC engine presented in Figure~\ref{fig:2}.

\medskip\textit{Initialization}: The initial state of the engine is assumed to be of product form
\begin{equation}
    \sigma^{EXY}_{i} \coloneqq \psi^{RS}\otimes \tau^{A}_{c} \otimes \tau^{B_{h}}_{h} \otimes \tau^{B_{c}}_{c} \otimes |0\rangle\!\langle 0|^{X}\otimes |0\rangle\!\langle 0|^{Y}\;, \label{eqn:initial}
\end{equation}
where $\psi^{RS}=|\psi\rangle\!\langle \psi|^{RS}$ is the purification of the arbitrary initial system state $\sigma^{S}_{i}$, $|0\rangle\!\langle 0|^{X}$ and $|0\rangle\!\langle 0|^{Y}$ are the initial states of the classical registers $X$ and $Y$, respectively, and $\tau^{A}_{c}$, $\tau^{B_{h}}_{h}$, and  $\tau^{B_{c}}_{c}$ are thermal states of $A$, $B_{h}$, and $B_{c}$ with temperatures $T_{c}$, $T_{h}$, and $T_{c}$, respectively.

\medskip\textit{Initial measurement stage}: Here we prepare $A$ for the encoding stage by making a projective measurement $\left\{P^{A}_{x}\right\}_{x\in \mathcal{X}}$ onto its energy basis (i.e., $P^{A}_{x}=|\varepsilon_{x}\rangle\!\langle \varepsilon_{x}|^{A}$), and hence the non-selective, post-measurement state is as follows:
\begin{equation}
    \sigma^{EXY}_{0} \coloneqq \sum_{x}p_{X}(x)\frac{P^{A}_{x}\sigma^{E}_{i}P^{A}_{x}}{p_{X}(x)}\otimes |x\rangle\!\langle x|^{X}\otimes |0\rangle\!\langle 0|^{Y}\;,
\end{equation}
where
\begin{equation}
    \frac{P^{A}_{x}\sigma^{E}_{i}P^{A}_{x}}{p_{X}(x)}=|\psi\rangle\!\langle \psi|^{RS} \otimes   |\varepsilon_{x}\rangle\!\langle \varepsilon_{x}|^{A}\otimes \tau^{B_{h}}_{h} \otimes \tau^{B_{c}}_{c}\;.
\end{equation}

\medskip\textit{Adaptive encoding stage}: Using the measurement outcome $x\in \mathcal{X}$, we encode the system $RS$ using the ancillary system $A$ and a unitary $U^{SA}_{\text{enc}, x}$ as follows:
\begin{equation}
    \sigma^{RSA}_{\text{enc},x} \coloneqq U^{SA}_{\text{enc}, x}(|\psi\rangle\!\langle \psi|^{RS}\otimes |\varepsilon_{x}\rangle\!\langle \varepsilon_{x}|^{A})(U^{SA}_{\text{enc}, x})^{\dagger}\;, \label{eqn:post_enc}
\end{equation}
and the total state $\sigma^{EXY}_{\text{enc}}$ becomes
\begin{equation}
     \sum_{x}p_{X}(x)  \sigma^{RSA}_{\text{enc},x} \otimes \tau^{B_{h}}_{h} \otimes \tau^{B_{c}}_{c}\otimes |x\rangle\!\langle x|^{X}\otimes |0\rangle\!\langle 0|^{Y}\;. \label{eqn:sig_enc}
\end{equation}

\medskip\textit{Thermal operation stage}: The encoded system $RSA$ interacts with the hot bath via the unitary $U^{SAB_{h}}$ and leads to the following state:
\begin{equation}
    \sigma^{RSAB_{h}}_{1,x} \coloneqq U^{SAB_{h}}(\sigma^{RSA}_{\text{enc},x}\otimes \tau^{B_{h}}_{h})(U^{SAB_{h}})^{\dagger}\;.
\end{equation}
Hence $\sigma^{EXY}_{1}$ at this stage is given by
\begin{equation}
   \sum_{x}p_{X}(x) \sigma^{RSAB_{h}}_{1,x}\otimes \tau^{B_{c}}_{c}\otimes |x\rangle\!\langle x|^{X} \otimes |0\rangle\!\langle 0|^{Y}\;. \label{eqn:sig_1}
\end{equation}
It is important to note that, in general, the noisy channel does not act on $S$ and $A$ independently. This relaxes one of the assumptions often made in the QEC literature, namely that of independent noisy channels acting on individual subsystems (specifically qubits).

\medskip\textit{Error-syndrome measurement stage}: After the thermal operation stage, we need a decoding process to recover the initial state of the system $S$. This process can be divided into two stages: error-syndrome measurement and error correction, characterised by the ($x$-value dependent) quantum instrument $\{ \mathcal{N}^{SA}_{y|x}\}_{y}$ and unitary operators $\{ U^{SA}_{\text{dec}, xy}\}_{x,y}$, respectively, where $y \in \mathcal{Y}_x$ is the measurement outcome of the former stage (we write $\mathcal{Y}_x$ because the set of possible outcomes is generally a function of the initial measurement outcome $x$). Therefore, the former is picked depending on the initial measurement outcome $x \in \mathcal{X}$, and the latter can generally depend on both measurement outcomes $(x,y)\in (\mathcal{X}, \mathcal{Y}_x)$. The state of the engine after the error-syndrome measurement step is described by the quantum-classical state (using the definition in Eq.~\eqref{eqn:quantum_instrument})
\begin{equation}
    \sigma^{EXY}_{2} \coloneqq \sum_{x,y}p_{XY}(x,y)\sigma^{E}_{2,xy}\otimes |x,y\rangle\!\langle x,y|^{XY}\;,
\end{equation}
where $p_{XY}(x,y)=p_{X}(x)p_{Y|X}(y|x)$ and
\begin{equation}
    \sigma^{E}_{2,xy} \coloneqq \frac{\mathcal{N}^{SA}_{y|x}(\sigma^{RSAB_{h}}_{1,x})}{p_{Y|X}(y|x)} \otimes \tau^{B_{c}}_{c}\;.
\end{equation}
It is important to note that the system $Y$ is the classical register used to record the measurement outcome, not the measurement apparatus itself. The latter is implicit in the quantum instrument $\{ \mathcal{N}^{SA}_{y|x}\}_{y}$ \cite{nielsen1998information, hayashi2006quantum}, which makes this stage non-unitary in general. Let us use the notation
\begin{equation}
    \mathcal{N}^{SA\rightarrow SAY}_{|x}(\cdot)\coloneqq \sum_{y\in \mathcal{Y}_{x}}\mathcal{N}^{SA}_{y|x}(\cdot)\otimes |y\rangle \! \langle y|^{Y}\;, \label{eqn:instrument_notation}
\end{equation}
to indicate the joint action of the instrument on the system and the memory.

\medskip\textit{Error-correction stage}: The decoding is completed by selectively applying $U^{SA}_{\text{dec},xy}$ based on the measurement outcomes $(x,y)$ as
\begin{equation}
    \sigma^{EXY}_{3}=\sum_{x,y}p_{XY}(x,y)\sigma^{E}_{3,xy}\otimes |x,y\rangle\!\langle x,y|^{XY}\;,
\end{equation}
where 
\begin{equation}
    \sigma^{E}_{3,xy}=U^{SA}_{\text{dec},xy}\sigma^{E}_{2,xy}(U^{SA}_{\text{dec},xy})^{\dagger}\;.
\end{equation}

\medskip\textit{Resetting of the ancillary system}: To prepare the ancillary system for a future cycle, we let $A$ interact with the cold bath to return it to its initial thermal state. In order to disentangle $A$ from the rest of the engine, we pick the thermal operation at this stage to be a fully thermalizing channel, which simply replaces the final state of $A$ with a newly prepared copy of $\tau^{A}_{c}$. The unitary that realizes this process evolves the state via
\begin{equation}
    \sigma^{EXY}_{f}=U^{AB_{c}}\sigma^{EXY}_{3}(U^{AB_{c}})^{\dagger}\;.
\end{equation}
We note that this can be realized by taking $\sigma^{B_{c}}_{3}=\tau^{B_{c}}_{c}=\otimes^{n}_{i=1}\tau^{B^{(i)}_{c}}_{c}$ (for some integer $n$ describing the size of the bath) and applying a SWAP operation (defined by $U_{\text{SWAP}}|\psi\rangle \otimes |\phi\rangle = |\phi \rangle \otimes |\psi\rangle$ for any two states $|\psi\rangle$ and $|\phi\rangle$) between $\sigma^{A}_{3}$ and $\tau^{B_{c}^{(1)}}_{c}$, leading to $\sigma_{f}^{A}=\tau^{A}_{c}$ and $\sigma^{B_{c}}_{f}=\sigma^{A}_{3}\otimes^{n}_{i=2}\tau^{B^{(i)}_{c}}_{c}$.

\medskip\textit{Discarding the classical registers}: At the final stage, we discard the classical registers $X$ and $Y$. As we discussed after Eq.~\eqref{eqn:cond_ent}, this discarding process is equivalent to an erasure of the Maxwell demon's memory, which is important for completing the thermodynamic cycle.

\subsection{Information-Theoretic Analysis}
\label{sub:info-theo}

In this subsection, we compute the entropy change at each stage of operation, with the purpose of arriving at a formula for the net entropy change of the engine after discarding the classical registers $X$ and $Y$. All information-theoretic quantities are computed for the sequence of states with subscripts
\[ i \to 0 \to \text{enc} \to 1 \to 2 \to 3 \to f\;, \]

We start by noting that $H(EXY)_{\sigma_{0}}=H(EXY)_{\sigma_{i}}$ because $\sigma_{0}$ arises from $\sigma_{i}$ by means of a controlled unitary, where $A$ is the control system and $X$ is the target system. The controlled unitary is given by its action
\begin{equation}
|\varepsilon_{x}\rangle^{A}\otimes |0\rangle^{X}\longrightarrow |\varepsilon_{x}\rangle^{A}\otimes |x\rangle^{X}\;,
\end{equation}
%
via $V^{X}_{x}|0\rangle^{X}=|x\rangle^{X}$, and thus it can be written as
\begin{equation}
    \sum_{x}|\varepsilon_{x}\rangle\!\langle \varepsilon_{x}|^{A}\otimes V^{X}_{x} .
\end{equation}
The entropy at the encoding stage (defined by Eq.~\eqref{eqn:sig_enc}) also does not change $H(EXY)_{\sigma_{\text{enc}}}=H(EXY)_{\sigma_0}$ because this stage is also described by a controlled unitary where $X$ is the control system and $RSA$ is the target system. Moreover, it is clear from Eqs.~\eqref{eqn:sig_enc} and \eqref{eqn:sig_1} that $H(EXY)_{\sigma_{1}}=H(EXY)_{\sigma_{\text{enc}}}$ due to the unitary interaction between $RSA$ and $B_{h}$. 

The first non-unitary step is the error-syndrome measurement stage. The entropy after the measurement is computed as follows:
\begin{align}
    H(EXY)_{\sigma_{2}} &= H(X) + H(EY|X)_{\sigma_{2}} \\
    &= H(X)+\sum_{x}p_{X}(x)H(EY)_{\sigma_{2,x}}\;, \label{eqn:error_syndrome_ent}
\end{align}
where we have used the notation $H(EY|X=x)_{\sigma_{2}}\equiv H(EY)_{\sigma_{2,x}}$. Note that it is convenient to isolate the contributions due to $X$ because $X$ is the system conditioned on which the measurement on the engine is chosen. We continue writing $H(EY)_{\sigma_{2,x}}$ in the form
\begin{align}
& H(EY)_{\sigma_{2,x}} \notag \\
   &= H(E)_{\sigma_{1,x}}-\left( H(E)_{\sigma_{1,x}}-H(EY)_{\sigma_{2,x}} \right) \notag  \\
    & = H(E)_{\sigma_{1,x}} - \left( H(E)_{\sigma_{1,x}}-H(Y)_{\sigma_{2,x}}-H(E|Y)_{\sigma_{2,x}} \right) \notag \\
    & = H(E)_{\sigma_{1,x}}+H(Y)_{\sigma_{2,x}}-I_{G}(\{ \mathcal{N}^{SA}_{y|x} \}_{y}; \sigma^{E}_{1,x})\;,
\end{align}
where in the last line we have used the definition of the Groenewold information gain in Eq.~\eqref{eqn:groenewold_def}. From now on, we shall denote the Groenewold information gain $I_{G}(\{ \mathcal{N}^{SA}_{y|x} \}_{y}; \sigma^{E}_{1,x})$ by $I_{G,x}$ for brevity. Substituting this result back into Eq.~\eqref{eqn:error_syndrome_ent} yields
\begin{align}
   &H(X)+\sum_{x}p_{X}(x)\left(H(E)_{\sigma_{1,x}}+H(Y)_{\sigma_{2,x}}-I_{G,x}\right)\notag \\
   &=H(X)+H(E|X)_{\sigma_{1}}+H(Y|X)_{\sigma_{2}}-\sum_{x}p_{X}(x)I_{G,x} \notag \\
    &= H(EXY)_{\sigma_{1}}+H(Y|X)_{\sigma_{2}}-\sum_{x}p_{X}(x)I_{G,x}\;,
\end{align}
where in the last equality we used the fact that $H(EXY)_{\sigma_{1}}=H(EX)_{\sigma_{1}}$. Thus, the entropy change in the error-syndrome measurement stage becomes
\begin{equation}
H(EXY)_{\sigma_{2}}-H(EXY)_{\sigma_{1}}=    H(Y|X)_{\sigma_{2}}-\sum_{x}p_{X}(x)I_{G,x}\;. \label{eqn:non_unitary_ent}
\end{equation}
Regarding the error-correction stage, we can easily see that $H(EXY)_{\sigma_{3}}=H(EXY)_{\sigma_{2}}$ because this stage is also described by a controlled unitary. Finally, we obviously have $H(EXY)_{\sigma_{f}}=H(EXY)_{\sigma_{3}}$ for the recycling stage of the ancillary system due to the unitarity of the SWAP operation between $A$ and an identical subsystem from~$B_{c}$.

The final step of the QEC cycle is the second, and last non-unitary step after the error-syndrome measurement: here we discard the classical registers $X$ and $Y$. The entropy change due to the discarding process is simply given by $H(E)_{\sigma_{f}}-H(EXY)_{\sigma_{f}}=-H(XY|E)_{\sigma_{f}}$.
Namely, discarding the classical registers only changes the entropy of the engine if the information encoded in them cannot be perfectly recovered from the final reduced state of the engine $E$.

From the above analysis, we can easily compute the total entropy change of $E$ by using Eqs.~\eqref{eqn:initial} and \eqref{eqn:non_unitary_ent} and we summarize it in the following: 

\begin{theorem}
\label{th:info_theo}
The entropy change of the quantum error-correcting engine described in Figure~\ref{fig:2} after discarding the classical registers $X$ and $Y$ can be written in terms of information-theoretic quantities as
\begin{multline}
    H(E)_{\sigma_{f}}-H(E)_{\sigma_{i}}=\\H(Y|X)_{\sigma_{2}}
    -\sum_{x}p_{X}(x)I_{G,x}-H(XY|E)_{\sigma_{f}}\;.  
    \label{eqn:info-theo-eqn}
\end{multline}
\end{theorem}

\subsection{Thermodynamic Analysis}
\label{sub:thermo}

Here we derive a formula for the entropy change $H(E)_{\sigma_{f}}-H(E)_{\sigma_{i}}$ in terms of the heat dissipated into the two baths (which can be measured as described in, e.g., \cite{goold2015nonequilibrium}). If we denote by $\mathscr{H}^{B_{h(c)}}$ the Hamiltonian of the hot (cold) bath and $Q_{B_{h(c)}}=\langle \mathscr{H}^{B_{h(c)}}\rangle _{i}-\langle \mathscr{H}^{B_{h(c)}} \rangle _{f}$ the total heat absorbed by $SA$ from the hot (cold) bath over the whole QEC cycle (where $\langle \cdots \rangle_{i}$ and $\langle \cdots \rangle_{f}$ denote expectations over the initial and final states of the baths, respectively), then we can formulate the following theorem to describe the entropy change in terms of the dissipated heat.

\begin{theorem}\label{th:entropy-change}
	The entropy change of the quantum error-correcting engine described in Figure~\ref{fig:2} after discarding the classical registers $X$ and $Y$ can be written in terms of the heat dissipated into the hot and cold baths as
    \begin{multline}
    H(E)_{\sigma_{f}}-H(E)_{\sigma_{i}}=S_{e}-\frac{Q_{h}}{k_{B}T_{h}}-\frac{Q_{c}}{k_{B}T_{c}} \\ -I(RS:B_{h}B_{c})_{\sigma_{f}}-\Gamma(B_{h}B_{c})_{\sigma_{f}}\;,  \label{eqn:thermoeqn}  
\end{multline}
where 
\begin{multline}
    \Gamma(B_{h}B_{c})_{\sigma_{f}} \coloneqq I(B_{h}:B_{c})_{\sigma_{f}}+D(\sigma^{B_{h}}_{f}\Vert\tau^{B_{h}}_{h})\\+D(\sigma^{B_{c}}_{f}\Vert\tau^{B_{c}}_{c}) \ge 0\;,
\end{multline}
depends only on the final states of the thermal baths and $S_{e}=H(RS)_{\sigma_{f}}$ is the entropy exchange.
\end{theorem}

\begin{proof}
We start from the final entropy of the engine 
\begin{align}
    H(E)_{\sigma_{f}} & \equiv H(RSAB_{h}B_{c})_{\sigma_{f}}\\
    & =H(RSB_{h}B_{c})_{\sigma_{f}}+H(A)_{\sigma_{f}}\;,
\end{align}
where we have used the fact that the recycling of the ancillary system disentangles it from the rest of the engine, so that $I(RSAB_{h}B_{c}:A)_{\sigma_{f}}=0$ \cite{horodecki2003entanglement}. Furthermore, using the definition of quantum mutual information, we can simplify $H(RSB_{h}B_{c})_{\sigma_{f}}$ further to
\begin{multline}
    H(RSB_{h}B_{c})_{\sigma_{f}} =  H(RS)_{\sigma_{f}}+H(B_{h})_{\sigma_{f}}+H(B_{c})_{\sigma_{f}}\\-I(B_{h}:B_{c})_{\sigma_{f}}-I(RS:B_{h}B_{c})_{\sigma_{f}}\;.
\end{multline}
Hence, for the entropy difference, we have 
\begin{multline}
    H(E)_{\sigma_{f}}-H(E)_{\sigma_{i}}=H(RS)_{\sigma_{f}}+\Delta H(B_{h})_{\sigma_{f}} \\ +\Delta H(B_{c})_{\sigma_{f}}-I(RS:B_{h}B_{c})_{\sigma_{f}}-I(B_{h}:B_{c})_{\sigma_{f}}\;,
\end{multline}
where we have used Eq.~\eqref{eqn:initial} and the fact that $\Delta H(RS)=H(RS)_{\sigma_{f}}$ because $\sigma^{RS}_{i}=|\psi\rangle\!\langle \psi|^{RS}$ is pure. Finally, we can use the Reeb--Wolf formula Eq.~\eqref{eqn:landauer_eqn} for the entropy changes of the two baths, yielding
\begin{multline}
    H(E)_{\sigma_{f}}-H(E)_{\sigma_{i}}=H(RS)_{\sigma_{f}}-\frac{Q_{h}}{k_{B}T_{h}}-\frac{Q_{c}}{k_{B}T_{c}}\\-D(\sigma_{f}^{B_{h}}\Vert \tau^{B_{h}}_{h})-D(\sigma_{f}^{B_{c}}\Vert \tau^{B_{c}}_{c})\\-I(RS:B_{h}B_{c})_{\sigma_{f}}-I(B_{h}:B_{c})_{\sigma_{f}}\;,
\end{multline}
which proves the theorem.
\end{proof}

The entropy exchange $S_{e}$, introduced in Refs.~\cite{schumacher1996quantum, schumacher1996sending}, quantifies the purity of the final (error-corrected) system state $\sigma^{RS}_{f}$. Note that having $S_{e}=0$ is only necessary, but not sufficient for perfect QEC. More precisely, the entropy exchange is related to the entanglement fidelity~\eqref{eqn:fidelity} via the quantum Fano inequality, that is, \cite{schumacher1996sending}
\begin{equation}
    H(F_{e})+(1-F_{e})\ln(d^{2}-1)\ge S_{e}\;, \label{eqn:fano}
\end{equation}
where $H(x)\coloneqq-x\log x-(1-x)\log(1-x)$ denotes the binary entropy and $d=\dim (\mathcal{H}_{S}\otimes \mathcal{H}_{R})$. We can clearly see that $F_{e}=1 \Rightarrow S_{e}=0$, but the converse statement does not necessarily hold. 

\subsection{Second-Law Inequality}

Using the two forms of the entropy change of $E$ from Eqs.~\eqref{eqn:info-theo-eqn} and \eqref{eqn:thermoeqn}, we arrive at the total entropic balance equation for the QEC engine:
\begin{multline}
    \frac{Q_{h}}{k_{B}T_{h}}+\frac{Q_{c}}{k_{B}T_{c}}= S_{e}+\sum_{x}p_{X}(x)I_{G,x}\\ -H(Y|X)_{\sigma_{2}} +H(XY|E)_{\sigma_{f}}\\-\Gamma(B_{h}B_{c})_{\sigma_{f}}
    -I(RS:B_{h}B_{c})_{\sigma_{f}}\;. 
\end{multline}
The corresponding second-law inequality follows from $I(RS:B_{h}B_{c})_{\sigma_{f}}\ge 0$ and $\Gamma(B_{h}B_{c})_{\sigma_{f}} \ge 0$ and is stated as follows:

\begin{corollary}
Given a general quantum error-correcting engine described in Figure~\ref{fig:2}, the Clausius formulation of the second law takes the form
\begin{multline}
    \frac{Q_{h}}{k_{B}T_{h}}+\frac{Q_{c}}{k_{B}T_{c}}\leq S_{e}+\sum_{x}p_{X}(x)I_{G,x}\\-H(Y|X)_{\sigma_{2}}+H(XY|E)_{\sigma_{f}}\;, \label{eqn: 2nd_law}
\end{multline}
where the bound is achieved when $I(RS:B_{h}B_{c})_{\sigma_{f}}=\Gamma(B_{h}B_{c})_{\sigma_{f}}=0$ (see the statement of Theorem~\ref{th:entropy-change}).
\end{corollary}

Here, we remind the reader that a physical interpretation of each of the terms appearing in the second-law inequality has been provided above. This inequality adds upon previous literature \cite{landi2020thermodynamic, sagawa2008second} in various ways. We elaborate further on this point in Sec.~\ref{sec:discussion}. 

\section{Efficiency-Fidelity Trade-off}

In this section, we choose to focus on the thermodynamic efficiency of the QEC engine in Figure~\ref{fig:2}. This is motivated by the following observation: A heat engine operating cyclically (with respect to both the state and the Hamiltonian of the working fluid) between two heat baths can possess a thermodynamic efficiency $\eta$ that is no larger than the Carnot efficiency $\eta_{C}$ (in the absence of a feedback controller). However, the state-cyclicity condition in heat-engine thermodynamics can be reinterpreted as a perfect QEC condition of the working fluid. In other words, maximum work extraction in thermodynamics can be viewed as the optimization of the efficiency $\eta \rightarrow \eta^{\star}\equiv \max \eta$ of the engine (where $\eta^{\star}\equiv \eta_{c}$ with no feedback), given the cyclicity constraint $F_{e}=1$ in terms of the error-correction fidelity. Therefore, a natural extension to this observation is to ask whether the task of QEC can be framed in a similar way, namely, as an optimization of the error-correction fidelity $F_{e} \rightarrow F_{e}^{\star}\equiv \max F_{e}$, given any fixed value of the thermodynamic efficiency $\eta$ within its allowable range. Additionally, recall that the QEC engine in Figure~\ref{fig:2} is a heat engine with a feedback controller, and therefore, the thermodynamic efficiency $\eta$ can generally become larger than the Carnot efficiency $\eta_{C}$ \cite{sagawa2008second}. This means that optimizing $F^{\star}_{e}$ for a thermodynamically desirable efficiency $\eta\ge \eta_{C}$ is of interest.

In this section, we derive a direct trade-off relation between $F^{\star}_{e}$ and $\eta$ for a QEC engine operating in the super-Carnot regime ($\eta \ge \eta_{C}$). The $(F^{\star}_{e}, \eta)$ trade-off relation is a special case of a triple trade-off relation between the maximum achievable entanglement fidelity $F_{e}^{\star}$, the thermodynamic efficiency $\eta$ of the QEC engine, and the efficacy of the error-syndrome measurement. The triple trade-off relation is derived under a set of physically motivated assumptions, the validity of which is further demonstrated in Sec.~\ref{sub:example}. As a consequence, a second trade-off relation exists between the efficiency and the efficacy (defined in Eq.~\eqref{eqn:efficacy}) whenever the state-cyclicity condition $F_{e}=1$ (i.e., perfect QEC) is satisfied.

\subsection{Origins of the Trade-off: Quantum Fano Inequality}
\label{sub:fano}

In Sec.~\ref{sub:thermo}, we discussed the quantum Fano inequality. This inequality can be rearranged in the form
\begin{equation}
    H(F_{e}) \ge \ln(d^{2}-1)(F_{e}-1)+S_{e}\;,
\end{equation}
where the binary entropy on the left-hand side is a concave function in the region $F_{e}\in [0, 1]$, taking values between 0 (for $F_{e} \in \{0, 1\}$) and 1 (for $F_{e}=0.5$), and the right-hand side is a linear function of $F_{e}$ with a positive slope. A simple pictorial interpretation of the above inequality shows that for varying values of entropy exchange $S_{e} \ge 0$, there exists a direct trade-off relation between the exact value of $S_{e}$ and the upper value $F^{\star}_{e}$ of the entanglement fidelity $F_{e}\in [0, F^{\star}_{e}]$ satisfying the quantum Fano inequality for a fixed $S_{e}$. This observation constitutes the basis of the trade-off theorem that we derive in this section. 

\subsection{Definition of Thermodynamic Efficiency of a QEC engine}
Here we define what we mean by the thermodynamic efficiency of the QEC engine, noting that this requires a careful accounting of the input heat from the measurement apparatus, which we have considered in Sec.~\ref{sub:meas_heat}. Such a problem can also be related to the problem of determining the thermodynamic efficiency of a heat engine operating between multiple heat baths (here the measurement apparatus plays the role of the third bath). 

We define the thermodynamic efficiency of our QEC engine in Figure~\ref{fig:2} as $\eta\coloneqq -W_{\text{tot}}/Q_{\text{input}}$, where $W_{\text{tot}}$ is the total work output of the engine and $Q_{\text{input}}$ is the net positive heat that is absorbed by the ``working fluid'' $SA$. In what follows, we find an explicit form of $Q_{\text{input}}$. To do so, we start by noting that the net change in the internal energy of the working fluid $SA$ is given by
\begin{equation}
    \Delta \mathscr{U}^{SA}=\Delta \mathscr{U}^{RSA}= \Delta \mathscr{U}^{RSAY},
\end{equation}
where the first equality follows from the fact that $\sigma^{R}_{f}=\sigma^{R}_{i}$ (Please see the note after Figure~\ref{fig:2}), whereas the second equality follows from the fact that the memory of $Y$ is erased after the QEC operation. Note that, although it is not obvious how to compute the work and heat contributions of $\Delta \mathscr{U}^{SA}$ of the working fluid $SA$ alone \cite{jacobs2009second, bedingham2016thermodynamic, abdelkhalek2016fundamental, faist2015minimal, kwon2019three, hayashi2017measurement}, $\Delta \mathscr{U}^{RSAY}$ does not suffer from the same problem. In fact, based on the stages described in Sec.~\ref{sub:stages}, we have
\begin{align}
    \mathscr{U}^{RSAY}_{0}-\mathscr{U}^{RSAY}_{i} &=0\; , \\
    \mathscr{U}^{RSAY}_{\text{enc}}-\mathscr{U}^{RSAY}_{0} &= W_{\text{enc}}\; , \\
    \mathscr{U}^{RSAY}_{1}-\mathscr{U}^{RSAY}_{\text{enc}} &= Q_{h}\; ,
\end{align}
where we have used the definitions of heat and work resulting from Eq.~\eqref{eqn:work_heat}. For the error-syndrome measurement stage, we have
\begin{align}
     \mathscr{U}^{RSAY}_{2}-\mathscr{U}^{RSAY}_{1} = \mathrm{Tr}\left[  (\sigma^{RSAY}_{2}-\sigma^{RSAY}_{1})\mathscr{H}^{RSAY}\right]\; ,
\end{align}
where $\mathscr{H}^{RSAY}=\mathscr{H}^{RSAY}_{t_{1}}=\mathscr{H}^{RSAY}_{t_{2}}$ is the total Hamiltonian of $RSAY$ at the beginning and end of the error-syndrome measurement stage (at times $t_{1}$ and $t_{2}$, respectively). 
Recalling the definition of heat in Eq.~\eqref{eqn:heat_def}, we arrive at $\mathscr{U}^{RSAY}_{2}-\mathscr{U}^{RSAY}_{1}=W_{\text{meas}}+Q_{\text{meas}}$, with $Q_{\text{meas}}=\langle \mathscr{H}^{M}\rangle_{1}-\langle \mathscr{H}^{M}\rangle_{2}$, where $\mathscr{H}^{M}$ is the Hamiltonian of the measurement apparatus. Finally, for the last three stages of operation of the QEC engine, we have
\begin{align}
    \mathscr{U}^{RSAY}_{3}-\mathscr{U}^{RSAY}_{2} &= W_{\text{dec}}\; , \\
    \mathscr{U}^{RSAY}_{f}-\mathscr{U}^{RSAY}_{3} &= Q_{c}\; ,\\
    \mathscr{U}^{RSAY}_{\text{erase}}-\mathscr{U}^{RSAY}_{f} &= \langle \mathscr{H}^{Y} \rangle_{i}-\langle \mathscr{H}^{Y} \rangle_{f}\equiv Q^{Y}_{\text{erase}}\; ,
\end{align}
where $\mathscr{H}^{Y}$ is the Hamiltonian of the classical register~$Y$ and $\langle \cdots \rangle_{f}$ denotes the expectation with respect to  the final reduced state of $Y$, i.e., $\sigma^{Y}_{f}=\mathrm{Tr}_{EX}[\sigma^{EXY}_{f}]=\sum_{x}p_{X}(x)\sigma_{f,x}^{Y}$, with $\sigma^{Y}_{f,x}=\sum_{y \in \mathcal{Y}_{x}}p_{Y|X}(y|x)|y\rangle \! \langle y|^{Y}$. Consequently, the above discussion yields
\begin{equation}
    \Delta \mathscr{U}^{RSAY}=W_{\text{tot}}+Q_{\text{tot}}\;, \label{eqn:internal_energy_decomp}
\end{equation}
where
\begin{align}
W_{\text{tot}} & \equiv W_{\text{enc}}+W_{\text{meas}}+W_{\text{dec}},
\\
Q_{\text{tot}} & \equiv Q_{h}+Q_{\text{meas}}+Q_{c}+Q^{Y}_{\text{erase}}
\end{align}
are the total work and heat separations of the internal energy of $RSAY$. Note that although the total dissipated heat during the erasure of both memory systems $X$ and $Y$ is given by $Q_{\text{erase}}=Q^{X}_{\text{erase}}+Q^{Y}_{\text{erase}}$, the term $Q^{X}_{\text{erase}}$ is irrelevant for the above thermodynamic analysis of $\Delta \mathscr{U}^{RSAY}$.

To identify $Q_{\text{input}}$, we further decompose $Q_{\text{meas}}= Q^{Y}_{\text{meas}}+Q^{RSA\vert Y}_{\text{meas}}-k_{B}T_{m}D(\sigma_{2}^{M}\Vert \tau^{M}_{m})$ based on our discussion in Sec.~\ref{sub:meas_heat}, where $\sigma_{2}^{M}$ is the post-measurement state of the measurement apparatus $M$, and $\tau^{M}_{m}$ is its initial thermal state, with temperature $T_{m}$. By definition, we have
\begin{align}
\frac{Q^{Y}_{\text{meas}}}{k_{B}T_{m}}
& = \Delta H(Y|X)-I(Y:M\vert X)_{\sigma_{2}}\\
& =H(Y|XM)_{\sigma_{2}} \ge 0,    
\end{align}
whereas $Q^{RSA\vert Y}_{\text{meas}}\leq 0$ is typically satisfied for any meaningful QEC engine. This is because a measurement process is supposed to transfer entropy (and hence information) from the measured system to the measurement apparatus. This, combined with $Q_{h} \ge 0$, $Q_{c}\leq 0$ (since we are interested in the heat-engine regime of operation), and $Q^{Y}_{\text{erase}}\leq 0$ (which follows from the fact that the memory $Y$ is initialized from, and erased back to, its ground state), leads to the decomposition of $Q_{\text{tot}}$ into positive $Q_{h}+Q^{Y}_{\text{meas}}$ and negative $Q^{RSA\vert Y}_{\text{meas}}+Q_{c}+Q^{Y}_{\text{erase}}$ parts. Therefore, by definition, we have $Q_{\text{input}}\coloneqq Q_{h}+Q^{Y}_{\text{meas}}$.

Next, we make a physical simplification in the error-correction setting by introducing the following assumption:

\begin{assumption} \label{assump-1}
We consider the regime of error correction with ``sufficiently high fidelity'', such that $\vert \Delta \mathscr{U}^{RSAY} \vert  \ll Q_{\operatorname{input}}$.
\end{assumption}

By ``sufficiently high fidelity'', we mean that the QEC fidelity $F_{e}\ge 1-\varepsilon$, with $\varepsilon \ll (Q_{\text{input}}/2\Vert \mathscr{H}^{S}\Vert_{\infty})^{2}$ with $\Vert \cdot \Vert_{\infty}$ denoting the infinity norm of an operator. This is justified by noting that $\Delta \mathscr{U}^{RSAY}=\Delta \mathscr{U}^{RS}$ along with
\begin{align}
    \vert \Delta \mathscr{U}^{RS} \vert &=\vert \mathrm{Tr}\left[(\sigma^{RS}_{f}-\sigma^{RS}_{i})\mathscr{H}^{RS}\right] \vert \\ & \leq \Vert \sigma^{RS}_{f}-\sigma^{RS}_{i} \Vert_{1} \cdot \Vert \mathscr{H}^{RS} \Vert_{\infty} \\ & \leq 2\Vert \mathscr{H}^{S} \Vert_{\infty} \sqrt{\varepsilon},
\end{align}
where the first inequality is a special case of H\"{o}lder's inequality and the second follows from the inequality $\frac{1}{2}\Vert \rho -\sigma \Vert_{1} \leq \sqrt{1-F(\rho, \sigma)}$, holding for every two density matrices.

Consequently, the thermodynamic efficiency of the QEC engine has the form $\eta\coloneqq-W_{\text{tot}}/Q_{\text{input}}\approx Q_{\text{tot}}/Q_{\text{input}}$, which follows from dividing both sides of Eq.~\eqref{eqn:internal_energy_decomp} by $Q_{\text{input}}$ and then using Assumption~\ref{assump-1}. Also, $\eta \leq 1$ follows from the definition of $Q_{\text{input}}$ as the positive part of $Q_{\text{tot}}$, appearing in Eq.~\eqref{eqn:internal_energy_decomp}.
\subsection{Useful Thermodynamic and Information-Theoretic Statements}
With the above definition of thermodynamic efficiency, let us now compute the important thermodynamic quantity $Q_{h}/k_{B}T_{h}+Q_{c}/k_{B}T_{c}$ appearing in Theorem~\ref{th:entropy-change}. From $\eta=Q_{\text{tot}}/Q_{\text{input}}$, it follows that 
\begin{equation}
    \vert Q_{c} \vert = (1-\eta)(Q_{h}+Q^{Y}_{\text{meas}})-\vert Q^{RSA\vert Y}_{\text{meas}}\vert -\vert Q^{Y}_{\text{erase}}\vert\;, 
\end{equation}
and therefore
\begin{multline}
    \frac{Q_{h}}{k_{B}T_{h}}+\frac{Q_{c}}{k_{B}T_{c}}=\frac{Q_{h}}{k_{B}T_{h}}\left( \frac{\eta-\eta_{C}}{1-\eta_{C}} \right)\\+\frac{\vert Q^{RSA\vert Y}_{\text{meas}}\vert +\vert Q^{Y}_{\text{erase}}\vert-(1-\eta)Q^{Y}_{\text{meas}}}{k_{B}T_{c}}\;, \label{eqn:thermo_quantity}
\end{multline}
where $\eta_{C}=1-T_{c}/T_{h}$ is the Carnot efficiency.

One more relation that we shall need is the upper bound on the Groenewold information gain derived in \cite{buscemi2016approximate}. This is given by (for a specific realization $x \in \mathcal{X}$ of the initial measurement stage)
\begin{equation}
    I_{G,x} \leq H(Y|X=x)_{\sigma_{2}}-D(\sigma^{E}_{1,x}\Vert \mathcal{N}^{\dagger}_{|x}\circ \mathcal{N}_{|x}(\sigma_{1,x}^{E}))\;, \label{eqn:BDW_upper}
\end{equation}
where we have used the notation $\mathcal{N}_{|x}\equiv \mathcal{N}^{SA\rightarrow SAY}_{|x}$ from Eq.~\eqref{eqn:instrument_notation}. The relative entropy term is not necessarily positive because the input $\mathcal{N}^{\dagger}_{|x}\circ \mathcal{N}_{|x}(\sigma_{1,x}^{E})$ is not a density matrix. This can be seen by normalizing the positive semi-definite operator $\mathcal{N}^{\dagger}_{|x}\circ \mathcal{N}_{|x}(\sigma_{1,x}^{E})$ by dividing by its trace, leading to \cite{buscemi2020thermodynamic}
\begin{multline}
    D(\sigma^{E}_{1,x}\Vert \mathcal{N}^{\dagger}_{|x}\circ \mathcal{N}_{|x}(\sigma_{1,x}^{E}))=D(\sigma^{E}_{1,x}\Vert \tilde{\sigma}^{E}_{1,x})\\-\ln \mathrm{Tr}\!\left[\mathcal{N}^{\dagger}_{|x}\circ \mathcal{N}_{|x}(\sigma^{E}_{1,x})\right]\;, 
\end{multline}
where the first term is now non-negative (because $\tilde{\sigma}^{E}_{1,x}$ is a density matrix) and the second term is the negative log efficacy $\mathscr{E}(\mathcal{N}_{|x}; \sigma^{E}_{1,x})$ of the quantum instrument $\mathcal{N}_{|x}$ with respect to  the density matrix $\sigma^{E}_{1,x}$, as defined in Eq.~\eqref{eqn:efficacy}. Using these relations, we arrive at
\begin{multline}
     H(Y|X)_{\sigma_{2}}-\sum_{x\in \mathcal{X}}p_{X}(x)I_{G,x}  \\ \ge \sum_{x\in \mathcal{X}}p_{X}(x)\left[-\ln \mathscr{E}(\mathcal{N}_{|x}; \sigma^{E}_{1,x})\right] \; . \label{eqn:info_theo_quantity}
\end{multline}

\subsection{Lower Bound on Entropy Exchange}

Due to the importance of entropy exchange $S_{e}$ in arriving at a trade-off relation (as emphasised in Sec.~\ref{sub:fano}), we express this quantity as a function of the thermodynamic efficiency $\eta$, defined above. To start, we make the following remark that is justified in an experimental setting, since access to cooling resources is typically limited.

\begin{remark} \label{remark1}
We treat the cold bath $B_{c}$ as a shared cooling resource among various subsystems in the QEC setting. That is, $B_{c}$ is not only used to recycle the ancillary system $A$, but it can also be used to erase the memory of the classical registers $X$ and $Y$.
\end{remark}

This constitutes the basis for the following assumption:

\begin{assumption}
\label{assump-2}
Pre-erasure, the entropy of the memory~$X$ is negligible with respect to  the entropy of $Y$, namely, $H(X|E)_{\sigma_{f}} \ll  H(Y|EX)_{\sigma_{f}}$.
\end{assumption}

Here we offer the following physical justification. Due to Remark~\ref{remark1}, we expect the shared cold bath (amongst systems $A$, $X$, and $Y$) to be sufficiently cold to accomplish the erasure task of the classical memories. Therefore, we expect the entropy of $X$ (which depends on the statistics of the initial measurement stage) to be negligible relative to that of $Y$, which is determined by the statistics of the error-syndrome measurement, the latter being rich in most realistic conditions.

From the above two assumptions, it follows that in Theorem~\ref{th:info_theo}, we have the approximation
\begin{align}
H(XY|E)_{\sigma_{f}}& =H(X|E)_{\sigma_{f}}+H(Y|EX)_{\sigma_{f}}\\
& \approx H(Y|EX)_{\sigma_{f}}.    
\end{align}
Moreover, we can further simplify $H(Y|EX)_{\sigma_{f}}$ as
\begin{align}
    H(Y|EX)_{\sigma_{f}} & = H(Y|RSAB_{h}B_{c}X)_{\sigma_{f}}\\
    & =H(Y|RSB_{c}X)_{\sigma_{f}}\\
    & =H(Y|RSAX)_{\sigma_{3}}
    \\& =H(Y|X)_{\sigma_{2}}-I(Y:RSA|X)_{\sigma_{3}}\;.
\end{align}
Given that we can achieve QEC with sufficiently high fidelity (Assumption~\ref{assump-1}), $\sigma^{RSA}_{3,x}$ should be approximately back to the post-encoded state $\sigma^{RSA}_{\text{enc},x}$. It follows that $I(Y:RSA|X)_{\sigma_{3}} \ll H(Y|X)_{\sigma_{2}}$, and hence $H(Y|EX)_{\sigma_{f}}\approx H(Y|X)_{\sigma_{2}}$. 
Finally, recall that the latter term also appears in Landauer's bound (e.g., see Eq.~\eqref{eqn:landauer_eqn}, for a given realization $x$) which puts a lower bound on the erasure heat of the memory $Y$, as $\vert Q^{Y}_{\text{erase}}\vert \ge k_{B}T_{c}H(Y|X)_{\sigma_{2}}$. This holds because the erasure process of $Y$, for a given value $x$ of the initial measurement stage, can be carried out unitarily. Averaging Eq.~\eqref{eqn:landauer_eqn} over the realizations $x\in \mathcal{X}$, we arrive at the desired inequality. Therefore, we make the following assumption.

\begin{assumption}
\label{assump-3}
The Landauer bound is reachable approximately in experiments, namely $\vert Q^{Y}_{\operatorname{erase}}\vert \approx  k_{B}T_{c}H(Y|X)_{\sigma_{2}}$.
\end{assumption}

From this, it follows that $H(XY|E)_{\sigma_{f}}\approx H(Y|X)_{\sigma_{2}} \approx \vert Q^{Y}_{\text{erase}}\vert/k_{B}T_{c}$. Substituting this into Theorem~\ref{th:info_theo}, we get $\Delta H(E)_{\sigma_{f}}\approx H(Y|X)_{\sigma_{2}}-\sum_{x\in \mathcal{X}}p_{X}(x)I_{G,x}-\vert Q^{Y}_{\text{erase}}\vert/k_{B}T_{c}$. This, along with Theorem~\ref{th:entropy-change} and Eq.~\eqref{eqn:thermo_quantity}, gives us
\begin{multline}
    S_{e}\ge \frac{Q_{h}}{k_{B}T_{h}}\left( \frac{\eta-\eta_{C}}{1-\eta_{C}} \right)+\frac{\vert Q^{RSA\vert Y}_{\text{meas}}\vert -(1-\eta)Q^{Y}_{\text{meas}}}{k_{B}T_{c}}\\+\left(H(Y|X)_{\sigma_{2}}-\sum_{x\in \mathcal{X}}p_{X}(x)I_{G,x}\right)\;.
\end{multline}
A simple manipulation via $(1-\eta)Q_{\text{meas}}^{Y}=(1-\eta_{C})Q_{\text{meas}}^{Y}-(\eta-\eta_{C})Q_{\text{meas}}^{Y}$ yields
\begin{multline}
    S_{e}\ge \frac{Q_{\text{input}}}{k_{B}T_{h}}\left( \frac{\eta-\eta_{C}}{1-\eta_{C}} \right)+\frac{\vert Q^{RSA\vert Y}_{\text{meas}}\vert -(1-\eta_{C})Q^{Y}_{\text{meas}}}{k_{B}T_{c}}\\+\left(H(Y|X)_{\sigma_{2}}-\sum_{x\in \mathcal{X}}p_{X}(x)I_{G,x}\right)\;. 
\end{multline}
The middle term on the right-hand side of the above equality becomes positive when $|Q^{RSA|Y}_{\text{meas}}|\ge \frac{T_{c}}{T_{h}}Q^{Y}_{\text{meas}}$, which implies that the measurement process leads to a ``sufficient'' dissipation of heat from the encoded system to the measurement apparatus.  Therefore, the last assumption we make is given as follows.

\begin{assumption}
\label{assump-4}
The thermodynamic constraint $|Q^{RSA|Y}_{\operatorname{meas}}|\ge \frac{T_{c}}{T_{h}}Q^{Y}_{\operatorname{meas}}$ applies to the measurement process. The equivalent information-theoretic constraint is given by $I(RSA:M|XY)_{\sigma_{2}}+\sum_{x}p_{X}(x)I_{G,x} \ge \frac{T_{c}}{T_{h}}H(Y|X)_{\sigma_{2}}$.
\end{assumption}

This assumption, along with Eq.~\eqref{eqn:info_theo_quantity}, yields the final inequality
\begin{equation}
    S_{e} \ge \frac{Q_{\text{input}}}{k_{B}T_{h}}\left( \frac{\eta-\eta_{C}}{1-\eta_{C}} \right) +\sum_{x\in \mathcal{X}}p_{X}(x)\left[-\ln \mathscr{E}(\mathcal{N}_{|x}; \sigma^{E}_{1,x})\right]\; . 
\end{equation}

Therefore, we summarize the main result of this section as
\begin{theorem}
\label{th:ent_ex}
    Given a quantum error-correcting engine shown in Figure~\ref{fig:2} and the set of physical Assumptions~\ref{assump-1}--\ref{assump-4} provided above, the entropy exchange is bounded from below by the thermodynamic efficiency of the engine and the average negative log efficacy of the general quantum measurement implemented, namely
    \begin{equation}
        S_{e} \ge  \frac{Q_{\operatorname{input}}}{k_{B}T_{h}}\left( \frac{\eta-\eta_{C}}{1-\eta_{C}} \right) +\sum_{x\in \mathcal{X}}p_{X}(x)\left[-\ln \mathscr{E}(\mathcal{N}_{|x}; \sigma^{E}_{1,x})\right] \;. \label{eqn:ent_ex_final2}
    \end{equation}
\end{theorem}

\subsection{Derivation of a Triple Trade-off Relation}
Finally, trade-off relations between entanglement fidelity, thermodynamic efficiency, and efficacy now appear as a direct consequence of the quantum Fano inequality discussion in Sec.~\ref{sub:fano} and Theorem~\ref{th:ent_ex}. We formulate this as the following triple trade-off relation:

\begin{corollary}
\label{th:three-way}
    Given a quantum error-correcting engine shown in Figure~\ref{fig:2} and the set of physical Assumptions~\ref{assump-1}--\ref{assump-4} provided above, there exists a triple trade-off between the maximum achievable error-correcting fidelity $F^{\star}_{e}$, the super-Carnot efficiency $\eta$ of this error-correcting engine, and the average negative log efficacy of the general quantum measurement, namely
    \begin{multline}
        H(F_{e})+(1-F_{e})\ln(d^{2}-1)\ge  \frac{Q_{\operatorname{input}}}{k_{B}T_{h}}\left( \frac{\eta-\eta_{C}}{1-\eta_{C}} \right) \\+\sum_{x\in \mathcal{X}}p_{X}(x)\left[-\ln \mathscr{E}(\mathcal{N}_{|x}; \sigma^{E}_{1,x})\right] \;. \label{eqn:three_way}
    \end{multline}
\end{corollary}
From here, two corollaries follow:
\begin{corollary} \label{th:trade_off}
    Given a quantum error-correcting engine shown in Figure~\ref{fig:2} and the set of physical Assumptions~\ref{assump-1}--\ref{assump-4} provided above, along with the assumption that the quantum instrument used has efficacy larger than one, there exists a direct trade-off between the maximum achievable error-correcting fidelity $F^{\star}_{e}$ and the super-Carnot efficiency $\eta$ of this error-correcting engine, namely
    \begin{equation}
        \eta > \eta_{C} \Rightarrow F_{e}^{\star}<1 \hspace{0.2cm} \text{or equivalently} \hspace{0.2cm}  F_{e}^{\star}=1 \Rightarrow \eta \leq \eta_{C}\;. \label{eqn:new_tradeoff}
    \end{equation}
\end{corollary}

\begin{corollary} \label{th:trade_off(2)}
    Given a quantum error-correcting engine shown in Figure~\ref{fig:2} and the set of physical Assumptions~\ref{assump-1}--\ref{assump-4} provided above, along with the assumption that the entanglement fidelity of the engine is equal to one, there exists a direct trade-off between the super-Carnot efficiency $\eta$ of the error-correcting engine, and the average negative log efficacy of the general quantum measurement, namely
    \begin{equation}
        \frac{Q_{\operatorname{input}}}{k_{B}T_{h}}\left( \frac{\eta-\eta_{C}}{1-\eta_{C}} \right) +\sum_{x\in \mathcal{X}}p_{X}(x)\left[-\ln \mathscr{E}(\mathcal{N}_{|x}; \sigma^{E}_{1,x})\right] \leq 0\;. \label{eqn:new_tradeoff(2)}
    \end{equation}
\end{corollary}
The latter implies that for all subunital quantum instruments $\mathcal{N}_{|x}$ (which includes efficient measurements), we must have $\eta \leq \eta_{C}$. This is in contrast with Ref.~\cite{sagawa2008second}. In fact, operating in the super-Carnot efficiency regime is still physically possible; however, it necessarily requires the use of superunital quantum instruments. Further comparison with Ref.~\cite{sagawa2008second} is made in Sec.~\ref{sec:discussion}.

\subsection{Examples}
\label{sub:example}

\subsubsection{Bit-Flip Channel}
Here we consider the simplest QEC code: the 3-qubit code corresponding to the bit-flip noise channel $\mathcal{T}(\cdot)=(1-p)(\cdot)+pX(\cdot)X$, where $0\leq p \leq 1$ is the bit-flip probability \cite{nielsen2002quantum}. This can model the so-called flip-flop interactions between spins. First, note that the bit-flip channel can be written as a thermal operation Eq.~\eqref{eqn:thermal_op}. This can be accomplished by having the energy spectrum of the hot bath to be comprised of two energies $\epsilon_{0}$ and $\epsilon_{1}$ (where the excited state can be highly degenerate), with projectors $P_{0}$ and $P_{1}$, respectively. Hence, we choose the thermal state of the bath to be $\tau^{B_{h}}_{h}=(1-p)P_{0}+pP_{1}$ and the unitary $U=I\otimes P_{0}+X\otimes P_{1}$.

The system of interest $S$ is a single qubit prepared in an arbitrary pure state $|\psi\rangle^{S}=a|0\rangle^{S}+b|1\rangle^{S}$ (hence there is no need for a purifying system $R$), and the ancillary system is two non-interacting qubits, each prepared in the thermal state $\tau^{A_{1}}_{c}=\pi_{1}|0\rangle \!\langle 0|^{A_{1}}+(1-\pi_{1})|1\rangle \!\langle 1|^{A_{1}}$ and $\tau^{A_{2}}_{c}=\pi_{2}|0\rangle \!\langle 0|^{A_{2}}+(1-\pi_{2})|1\rangle \!\langle 1|^{A_{2}}$, with Gibbs distributions $\{\pi_{1}, (1-\pi_{1}) \}$ and $\{ \pi_{2}, (1-\pi_{2}) \}$ respectively, having the same temperatures $T_{c}$ but slightly different energy separations (to avoid degeneracies in the combined ancillary system $A_{1}A_{2}$). 
The initial projective measurement of the ancillary system $A$ can have four different outcomes $\{00, 01, 10, 11 \}$ with probabilities $\{\pi_{1}\pi_{2}, \pi_{1}(1-\pi_{2}), (1-\pi_{1})\pi_{2}, (1-\pi_{1})(1-\pi_{2}) \}$, respectively. Therefore, the encoding stage starts by applying the unitaries $\{ I_{1}\otimes I_{2}, I_{1}\otimes X_{2}, X_{1}\otimes I_{2}, X_{1}\otimes X_{2} \}$ based on the corresponding measurement outcomes, so that we can arrive at the same post-measurement state $|00\rangle \! \langle 00|^{A}$. After this step, we complete the unitary encoding by applying two controlled CNOT gates (one on each ancillary qubit), leading to 
\begin{multline}
    (a|0\rangle^{S}+b|1\rangle^{S})\otimes |0\rangle^{A_{1}}\otimes |0\rangle^{A_{2}} \\ \rightarrow a|000\rangle^{SA}+b|111\rangle^{SA} \equiv |\Psi\rangle^{SA},
\end{multline}
which is independent of $x$. Next, we subject the encoded system to the bit-flip channel, acting independently on each qubit, namely $\mathcal{T}^{\otimes 3}(|\Psi\rangle \! \langle \Psi|^{SA})$. We follow this step by conducting the error-syndrome measurement, described by the four projectors 
\begin{align}
    &\begin{cases}
        &\Pi^{SA}_{y=0} =|000\rangle \! \langle 000|^{SA}+|111\rangle \! \langle 111|^{SA}\; \nonumber, \\ &\text{decoded by} \hspace{0.2cm} U^{SA}_{\text{dec}, y=0}=I^{SA}\; , 
    \end{cases} \\
    &\begin{cases}
        &\Pi^{SA}_{y=1} =|100\rangle \! \langle 100|^{SA}+|011\rangle \! \langle 011|^{SA}\;, \nonumber \\ &\text{decoded by} \hspace{0.2cm} U^{SA}_{\text{dec}, y=1}=X^{S}\otimes I^{A_{1}}\otimes I^{A_{2}}\; ,
    \end{cases} \\
    &\begin{cases}
        &\Pi^{SA}_{y=2} =|010\rangle \! \langle 010|^{SA}+|101\rangle \! \langle 101|^{SA}\;, \nonumber \\ &\text{decoded by} \hspace{0.2cm} U^{SA}_{\text{dec}, y=2}=I^{S}\otimes X^{A_{1}}\otimes I^{A_{2}}\; ,
    \end{cases} \\
    &\begin{cases}
        &\Pi^{SA}_{y=3} =|001\rangle \! \langle 001|^{SA}+|110\rangle \! \langle 110|^{SA}\;, \nonumber \\ &\text{decoded by} \hspace{0.2cm} U^{SA}_{\text{dec}, y=3}=I^{S}\otimes I^{A_{1}}\otimes X^{A_{2}}\; ,
    \end{cases}
\end{align}
corresponding to no bit-flip error, and a single bit-flip error in each qubit. Hence, the measurement statistics are given by $p_{Y|X}(y|x)=p_{Y}(y)=\mathrm{Tr}\left[\mathcal{T}^{\otimes 3}(|\Psi\rangle \! \langle \Psi|^{SA}) \Pi^{SA}_{y} \right]$ for $y=\{0, 1, 2, 3 \}$. Therefore, we require a memory system $Y$ with at least four orthonormal basis vectors $\{ |y\rangle \}^{3}_{y=0}$ (e.g., a system of two qubits).

In order to check the validity of the assumptions used to derive Theorem~\ref{th:three-way}, we need only check Assumptions~\ref{assump-2} and \ref{assump-4}, since Assumption~\ref{assump-1} depends on the strength of the noise present, Assumption~\ref{assump-3} has been validated experimentally \cite{berut2012experimental}, and the efficacy of the quantum instrument is equal to zero because the error-syndrome measurement is projective (and hence $\mathcal{N}_{|x}$ is unital). To check Assumption~\ref{assump-2}, we note that
\begin{align}
H(X|E)_{\sigma_{f}}& = H(X)_{\sigma_{f}} - I(X:E)_{\sigma_{f}}\\
& =H(X)_{\sigma_{f}},\\
H(Y|EX)_{\sigma_{f}}& = H(Y|E)_{\sigma_{f}}.    
\end{align}
Furthermore, we have
\begin{align}
H(X)_{\sigma_{f}} & = H(A)_{\tau_{c}},\\ H(Y|E)_{\sigma_{f}} & = H(Y)_{\sigma_{f}} - I(Y:SAB_{h}B_{c})_{\sigma_{f}}\\
& =H(Y)_{\sigma_{2}}-I(Y:SB_{c})_{\sigma_{f}}\\
& =H(Y)_{\sigma_{2}}-I(Y:SA)_{\sigma_{3}},     
\end{align}
which gives us an easy way to check Assumption~\ref{assump-2}, as $H(A)_{\tau_{c}}\ll H(Y)_{\sigma_{2}}-I(Y:SA)_{\sigma_{3}}$. For experimental convenience (controllability, ease of initialization, etc.), we assume that the qubit energies used are in the optical range, i.e., $\hbar \omega \sim 2-3$eV, and we assume the worst case scenario that the cold bath is at room temperature, hence $k_{B}T_{c}\approx 1/40$eV. Consequently, we have for $\hbar \omega/k_{B}T_{c} \sim 100$. This makes $H(A)_{\tau_{c}} \sim 10^{-42}$ (which is a function of the ratio $\hbar \omega /k_{B}T_{c}$), which we verify to be much smaller than $H(Y)_{\sigma_{2}}-I(Y:SA)_{\sigma_{3}}=H(SAY)_{\sigma_{3}}-H(SA)_{\sigma_{3}}\approx 0.166$ for a bit-flip error probability of $p=0.01$, when $a=b=1/\sqrt{2}$. Furthermore, Assumption~\ref{assump-4} is easily verified by checking that $I_{G}=H(Y|X)_{\sigma_{2}}=H(Y)_{\sigma_{2}}\approx 0.166$, hence $|Q^{RSA|Y}_{\text{meas}}|\ge \frac{T_{c}}{T_{h}}Q^{Y}_{\text{meas}}$ is true for any reasonable ratio $T_{c}/T_{h}$ of the QEC engine.

\subsubsection{Phase Damping Channel}

We would like to supplement the above example, which has a classical analogue, with the phase damping noise, described by a single parameter $\lambda$ as:
\begin{equation}
    \sigma=\begin{pmatrix} |a|^{2} & ab^{\star} \\ a^{\star}b & |b|^{2} \end{pmatrix} \rightarrow \begin{pmatrix} |a|^{2} & ab^{\star}e^{-\lambda} \\ a^{\star}be^{-\lambda} & |b|^{2} \end{pmatrix} \;. \label{eqn:phase_damp}
\end{equation}
Phase damping is a purely quantum-mechanical noise that appears in various implementations of qubits, e.g., nitrogen vacancy centers, where there is a large mismatch between transition energies of the qubit and that of the bath spins \cite{zhao2012decoherence, wang2012comparison} (hence flip-flop interactions are suppressed). One possible way to arrive at this noise via a thermal operation Eq.~($\ref{eqn:thermal_op}$) is to consider two systems, where the action of one on the other can be ignored (e.g., when one of the systems is macroscopic). Therefore, the larger system will be in local equilibrium, i.e., in the thermal state $\tau$. Furthermore, the joint Hamiltonian of the two systems can be written as $\mathscr{H}_{\text{eff}}(t)+\mathscr{H}^{B}$, where $\mathscr{H}_{\text{eff}}$ is an effective, stochastic Hamiltonian of the system. Averaging over all realizations of the stochastic noise, one can arrive at Eq.~($\ref{eqn:phase_damp}$), where $\lambda$ is generally a function of time.

It is known in the literature (see, e.g., \cite{nielsen2002quantum}) that a phase damping channel is equivalent to a phase-flip channel $\mathcal{T}(\cdot)=(1-p)(\cdot)+pZ(\cdot)Z$, where $0 \leq p \leq 1$ is the phase-flip probability. This equivalence follows from the unitary freedom in the choice of the Kraus operators of the phase damping channel, as
\begin{multline}
    \left\{\begin{pmatrix} 1 & 0 \\ 0 & \sqrt{1-\lambda} \end{pmatrix}\; , \begin{pmatrix}
    0 & 0 \\ 0 & \sqrt{\lambda}
    \end{pmatrix}\right\} \\ \rightarrow \left\{\sqrt{\mu}\begin{pmatrix} 1 & 0 \\ 0 & 1 \end{pmatrix}\; , \sqrt{1-\mu} \begin{pmatrix} 
    1 & 0 \\ 0 & -1
    \end{pmatrix}\right\}\;,
\end{multline}
where $\mu=(1-\sqrt{1-\lambda})/2$. Consequently, we need only correct the phase-flip channel. This can be done by encoding $S$ using two ancillary qubits $A=A_{1}A_{2}$  in the exact same way as for a bit-flip noise, with the difference that we include an additional step of applying a Hadamard gate to each of the three qubits \cite{nielsen2002quantum}, leading to the encoding
\begin{multline}
    (a|0\rangle^{S}+b|1\rangle^{S})\otimes |0\rangle^{A_{1}}\otimes |0\rangle^{A_{2}} \\ \rightarrow a|+++\rangle^{SA}+b|---\rangle^{SA} \equiv |\Phi\rangle^{SA}\;,
\end{multline}
where $|\pm \rangle =(|0\rangle \pm |1\rangle)/\sqrt{2}$. Therefore, we can continue in the exact same way as for the bit-flip case, but in the $|\pm \rangle$ basis, to arrive at the exact same results.

\section{Comparison with Previous Results}
\label{sec:discussion}

Our main results can be summarized by Eqs.~\eqref{eqn:result_groenewold}, \eqref{eqn: 2nd_law}, \eqref{eqn:three_way}.
At this stage, we would like to show that Eq.~\eqref{eqn: 2nd_law} generalizes a result of Sagawa and Ueda \cite{sagawa2008second}. First, note that the authors do not make the distinction between the main system $S$ and the ancilla system $A$ in their letter. As a consequence, there is no implementation of an initial projective measurement to prepare the ancillary system, and hence the classical register $X$ is not present. More importantly, the authors do not discard the classical registers, which is equivalent to refraining from the erasure of Maxwell demon's memory (a practice that has historically led to many contradictions in thermodynamics, as we discuss in the introduction). This leads to the absence of the important $H(XY|E)_{\sigma_{f}}$ term in the entropic balance leading to Eq.~\eqref{eqn: 2nd_law}. Making these changes in the derivation, we arrive at
\begin{equation}
    \frac{Q_{h}}{k_{B}T_{h}}+\frac{Q_{c}}{k_{B}T_{c}} \leq S_{e}+I_{G}\;, \label{eqn:sagawa}
\end{equation}
which is the main result of Ref.~\cite{sagawa2008second}, but with a few important generalizations: $(i)$ The initial state of the system is arbitrary in our setup (rather than thermal). $(ii)$ The quantum measurement stage is completely general and the information gain from the measurement is quantified by the Groenewold information gain; this is in contrast with the special case of efficient measurements made in Ref.~\cite{sagawa2008second}, where the information gain is quantified by the ``quantum-classical mutual information'', which is a special case of the Groenewold information gain for efficient measurements. 

A important attempt to derive the second law of QTD in the context of QEC has been made in Ref.~\cite{landi2020thermodynamic}. There, the authors arrive at a formula for the total entropy production of a QEC engine. However, we consider this derivation to be too restrictive, for the following reasons: First, the authors considered a system of qubits acted upon independently by generalized amplitude damping noise (which is a special type of thermal noise). Second, the dissipated heat, work, and entropy production during the error-syndrome measurement stage is neglected, whereas we show that it deserves careful consideration. Third, the ancillary system (also a collection of qubits) was initialized and recycled back to its ground state, implying the presence of a cold bath of temperature $T_{c}=0$. The latter implication is not only unphysical, but it also does not allow for the derivation of other forms of the second law that includes arbitrary values of the cold bath temperature, e.g., the Clausius formulation as Eq.~\eqref{eqn: 2nd_law} in our article. Incidentally, it is exactly this issue that is circumvented by initializing the ancillary system from a thermal state rather than its ground state, including the classical register $X$, and conducting the initial measurement stage before proceeding with the encoding, as seen in Figure~\ref{fig:2}. Finally, it is worth mentioning that, similar to Ref.~\cite{sagawa2008second}, the authors do not consider the entropic cost of erasing the memory of the classical registers in their final entropic balance.

\section{Conclusion and Open Questions}

In this article, we take a thermodynamic point of view with respect to general quantum measurements (described in the quantum instrument formalism) and quantum error correction. The former allowed us to arrive at an important bound on the measurement heat in terms of the Groenewold information gain. This bound grants a physical interpretation of negative values of the Groenewold information gain as the directionality of the dissipated heat that is physically allowed in a measurement process. 

Moreover, we consider the thermodynamic approach to quantum error correction, regarding it as a heat engine with a feedback controller. This approach allows us to derive the second-law inequality for a  QEC engine under very general conditions. As a manifestation of the second law in this setting, we show that a trade-off relation exists between the maximum achievable fidelity of the error-correction process and the attainability of super-Carnot efficiencies of a QEC engine.

An open question here is that of time, namely: how long does the QEC process take to be completed? Is there any similar trade-off relation between the time that it takes and the entanglement fidelity? This is directly related to the distinction between operating at maximum efficiency and maximum power for a heat engine. Another interesting question is the manifestation of the second law in the strong coupling regime between the system and the hot bath. It is well known that,
in this regime, there is no agreed upon definition of heat and internal energy \cite{kwon2019three, esposito2010entropy}, and so a consideration of various definitions and their outcomes can be insightful.

We would finally like to mention that our efficiency-fidelity trade-off shares some similarities with thermodynamic uncertainty relations (TURs) \cite{barato2015thermodynamic, pietzonka2016universal, pietzonka2018universal, miller2021thermodynamic}. This can serve as a second point of contact between thermodynamics and error correction.

\begin{acknowledgements}
The authors dedicate this paper to the memory of Jonathan P.~Dowling. The idea that led to this article was inspired by the last course that J.~P.~D.~taught on quantum computing. A.~D.~would like to thank Armen Allahverdyan and Giacomo Guarnieri for very useful discussions. This work was supported by
the U.S. Army Research Office through the U.S. MURI
Grant No.~W911NF-18-1-0218. F.~B.~acknowledges
support from MEXT Quantum Leap Flagship Program
(MEXT Q-LEAP) Grant No.~JPMXS0120319794 and JSPS
KAKENHI Grants 19H04066, 20K03746, and 21H05183. The latex package quantikz \cite{kay2018tutorial} has been used to generate the circuit plot.
\end{acknowledgements}

\bibliographystyle{unsrt}

\bibliography{references}

\begin{thebibliography}{10}
\bibitem{mynote1}
For example, we can consider the thermal state
  $\tau^{B}=\frac{1}{Z}\sum_{\alpha}e^{-\beta
  \epsilon_{\alpha}}|\alpha\rangle\! \langle \alpha|^{B}$ of a thermal bath as
  part of a larger environment $E=BB\cdots B$, comprised of $N$ identical
  subsystems, in the pure state
  $|\beta\rangle^{E}=\frac{1}{\sqrt{Z}}\sum_{\alpha}e^{-\frac{\beta
  \epsilon_{\alpha}}{2}}|\alpha \rangle^{\otimes N}$, hence
  $\tau^{B}=\mathrm{Tr}_{E \setminus B}[|\beta \rangle \! \langle \beta|^{E}]$.
  Consequently, a quantum process of the form $\rho^{A}_{i}\rightarrow
  \rho^{A}_{f}=\mathrm{Tr}_{B}[U(\rho^{A}_{i}\otimes \rho^{B}_{i})U^{\dagger}]$
  can be written in the Stinespring form $
  \rho^{A}_{f}=\mathrm{Tr}_{E}[(U\otimes V^{E\setminus B})(\rho^{A}_{i}\otimes
  |\beta \rangle \! \langle \beta|^{E})(U\otimes V^{E\setminus B})^{\dagger}]$,
  where $V^{E\setminus B}$ is an arbitrary unitary on the extended environment
  $E \setminus B$.
\end{thebibliography}


\appendix


\section{Work and Heat}
\label{sub:work_heat}

In this appendix, we define the thermodynamic quantities of work, internal energy, and heat in a general setting (both weak and strong coupling limits).

Consider two systems $A$ and $B$, and suppose that the collective system $AB$ evolves unitarily in the time interval $[0, t]$ via the joint unitary $V^{AB}_{t}$. From the conservation of the total von Neumann entropy under unitary evolution $H(A,B)_{\sigma_{t}}=H(A,B)_{\sigma_{0}}$, we find the following for the sum of the local entropy changes:
\begin{equation}
    \Delta H(A)+\Delta H(B)=\Delta I(A:B) \ge 0  \; , \label{eqn:total_ent}
\end{equation}
where we have used the definition of conditional entropy. We postpone the standard assumption of $I(A:B)_{\sigma_{0}}=0$ to arrive at more general results \cite{breuer2002theory, perarnau2015extractable,  paz2019dynamics}, since in practice one might not be able to initialize $AB$ in a product form. Next, we relate the entropy change of $B$ with  $D(\sigma^{B}_{t}\Vert\sigma^{B}_{0})$. This can be achieved as follows
\begin{align}
    \Delta H(B)&=-\mathrm{Tr}_{B}\left[\sigma^{B}_{t}\ln{\sigma^{B}_{t}}\right]+\mathrm{Tr}_{B}\left[\sigma^{B}_{0}\ln{\sigma^{B}_{0}}\right]\\
    &=-\mathrm{Tr}_{B}\left[\sigma^{B}_{t}\ln{\sigma^{B}_{t}}\right]+\mathrm{Tr}_{B}\left[\sigma^{B}_{t}\ln{\sigma^{B}_{0}}\right] \notag \\
    & \qquad -\mathrm{Tr}_{B}\left[\sigma^{B}_{t}\ln{\sigma^{B}_{0}}\right]+\mathrm{Tr}_{B}\left[\sigma^{B}_{0}\ln{\sigma^{B}_{0}}\right]\\
    &=\mathrm{Tr}_{B}\left[(\sigma^{B}_{0}-\sigma^{B}_{t})\ln{\sigma^{B}_{0}}\right]-D(\sigma^{B}_{t}\Vert\sigma^{B}_{0})\; . \label{eqn:eff_free_energy}
\end{align}
The operator $\ln \sigma^{B}_{0}$ can be interpreted as an effective unitless Hamiltonian of the system $B$ (i.e., $\beta \mathscr{H}^{B}_{\text{eff}}$ where $\beta$ plays the role of inverse temperature); as a consequence, $\beta^{-1}D(\sigma^{B}_{t}\Vert\sigma^{B}_{0})$ can be interpreted as the change in the free energy from the Gibbs state of $B$ with the effective Hamiltonian. Combining Eqs.~\eqref{eqn:total_ent} and \eqref{eqn:eff_free_energy}, we arrive at
\begin{multline}
    \mathrm{Tr}_{B}\left[(\sigma^{B}_{0}-\sigma^{B}_{t})\ln{\sigma^{B}_{0}}\right]=\\-\Delta H(A)+\Delta I(A:B)+D(\sigma^{B}_{t}\Vert\sigma^{B}_{0}) \label{eqn:heat_to_entropy}\; .
\end{multline}
Now we make an assumption about controllability, thereby operationally distinguishing system $A$ from $B$. Namely, we assume that system $A$ can be controlled by an external parameter with pre-determined dynamics $\left\{\lambda_{t}\right\}_{t\ge 0}$; this is in contrast with system $B$ which is assumed to be uncontrollable \cite{balian2007microphysics}. The assumption of controllability is further extended to the interaction between $A$ and $B$, since one might want to describe a measurement process or a thermalization process that has a coupling phase (with a measurement apparatus or a heat bath respectively) and a decoupling phase. This can be written using the interaction Hamiltonian $\mathscr{V}^{AB}$ that depends on a second external control parameter $\left\{ \nu_{t} \right\}_{t\ge 0}$ with pre-determined dynamics. The controllability assumptions are reflected in the total (time-dependent) Hamiltonian as 
\begin{equation}
    \mathscr{H}^{AB}_{t}=\mathscr{H}^{A}_{\lambda_{t}}+\mathscr{H}^{B}+\mathscr{V}^{AB}_{\nu_{t}}\;,
\end{equation}
with the interaction unitary given by
\begin{equation}
    V^{AB}_{t}=\mathcal{T}\exp\left\{-\frac{i}{\hbar}\int^{t}_{0}d\tau \mathscr{H}^{AB}_{\tau} \right\}\;,
\end{equation}
where $\mathcal{T}$ denotes time ordering. Due to the fact that the collective system $AB$ is isolated, the change in the total energy of the system is equated to (by definition) the work done on $AB$ through the control parameters $\lambda_{t}$ and $\nu_{t}$ as 
\begin{align}
    W^{AB}_{t} & =\int^{t}_{0}d\tau \mathrm{Tr}\left[\sigma^{AB}_{\tau}(\partial_{\tau}\mathscr{H}^{A}_{\tau}+\partial_{\tau}\mathscr{V}^{AB}_{\tau})\right]\\
    & =\mathrm{Tr}\left[\sigma^{AB}_{t}\mathscr{H}^{AB}_{t}\right]-\mathrm{Tr}\left[\sigma^{AB}_{0}\mathscr{H}^{AB}_{0}\right]\;.
\end{align}
In the weak coupling limit, this definition reduces to the local definition $W^{A}_{t}=\mathrm{Tr}\left[\sigma^{A}_{t}\mathscr{H}^{A}_{\lambda_{t}}\right]-\mathrm{Tr}\left[\sigma^{A}_{0}\mathscr{H}^{A}_{\lambda_{0}}\right]$ \cite{balian2007microphysics}, which we have used in the article.

In order to define what we mean by heat, we need to specify what we mean by the internal energy $\mathscr{U}^{A}_{t}$ of the system $A$, at time $t$. To this end, there have been various approaches in the literature: The first is that internal energy has to be a local quantity; i.e., it has to be represented in terms of the reduced density matrix of the system under consideration \cite{kwon2019three}. The second approach is that the definition of internal energy need not be a local quantity and the heat dissipated from $A$ should be equal to (by definition) the heat absorbed by $B$ \cite{esposito2010entropy}, i.e., $Q^{A}_{t}=-Q^{B}_{t}$ at all times. Both of these approaches can be used to derive the laws of thermodynamics and they become identical in the weak coupling limit. In this article, we adopt the second approach by using the same definition in \cite{esposito2010entropy} as
\begin{equation}
    \mathscr{U}^{A}_{t}=\mathrm{Tr}\left[\sigma^{AB}_{t}(\mathscr{H}^{A}_{\lambda_{t}}+\mathscr{H}^{AB}_{\nu_{t}})\right]\;, \label{eqn:int_energy}
\end{equation}
which is clearly non-local in $A$ and reduces to the well known (local) definition $\mathrm{Tr}_{A}\left[\sigma^{A}_{t}H^{A}_{\lambda_{t}}\right]$ in the weak coupling limit. The definition in Eq.~\eqref{eqn:int_energy} is motivated by the intuition that the concept of an \textit{internal} energy local to system $A$ is unclear if its interaction with a second system $B$ is not negligible. From Eq.~\eqref{eqn:int_energy}, we easily get the 1$^{\text{st}}$ law of thermodynamics
\begin{equation}
    \Delta \mathscr{U}^{A}=W^{AB}_{t}+\langle \mathscr{H}^{B}\rangle_{0}-\langle \mathscr{H}^{B}\rangle_{t}\;, \label{eqn:general_1law}
\end{equation}
where $Q^{A}_{t}=-Q_{t}^{B}\coloneqq\langle \mathscr{H}^{B}\rangle_{0}-\langle \mathscr{H}^{B}\rangle_{t}$ is interpreted as the heat absorbed by $A$ whenever $B$ is a heat bath.

We consider the important case when system $B$ is a heat bath. This is the most standard case in the literature, where we assume that system $B$ is initially in a thermal state $\sigma^{B}_{0}=\tau^{B}=\exp(-\beta \mathscr{H}^{B})/Z^{B}$. Then Eq.~\eqref{eqn:heat_to_entropy} gives 
\begin{equation}
    \beta \Delta \langle \mathscr{H}^{B}\rangle=-\Delta H(A)+\Delta I(A:B)+D(\sigma^{B}_{t}\Vert\tau^{B})\;, \label{eqn:landauer} 
\end{equation}
which, combined with the assumption of an initially product state $\sigma^{AB}_{0}=\sigma^{A}_{0}\otimes \tau^{B}$, leads to the well known Landauer bound $\beta Q_{t}^{B} \ge -\Delta H(A)$ \cite{reeb2014improved, allahverdyan2001breakdown, goold2015nonequilibrium}. System $A$ in this context is usually identified with the memory system of the Maxwell demon (not to be confused with the Maxwell demon itself, which typically represents the measurement apparatus). Using  Eqs.~\eqref{eqn:landauer} and \eqref{eqn:general_1law}, we arrive at
\begin{multline}
    \Delta \mathscr{U}^{A}=W^{AB}_{t}+ k_{B}T\Delta H(A)\\-k_{B}T\left(\Delta I(A:B)+D(\sigma^{B}_{t}\Vert\tau^{B})\right)\;.
\end{multline}
When the initial state of $AB$ is of product form, we immediately get the familiar 1$^{\text{st}}$ law inequality of non-equilibrium thermodynamics \cite{balian2007microphysics}
\begin{equation}
    \Delta \mathscr{U}^{A}\leq W^{AB}_{t}+k_{B}T\Delta H(A)\;, \label{eqn:1st_law_inequality}
\end{equation}
where the equality is reached for a thermodynamically ``reversible'' process described by the conditions $I(A:B)_{\sigma_{t}}=0$ and $D(\sigma^{B}_{t}\Vert\tau^{B})=0$. Furthermore, Thompson's formulation of the 2$^{\text{nd}}$ law \cite{allahverdyan2002thomson} follows almost immediately from the definition of free energy of $A$ as \cite{esposito2010entropy}
\begin{equation}
    F^{A}_{t}\coloneqq \mathscr{U}^{A}_{t}-k_{B}TH(A)_{\sigma_{t}}\;,
\end{equation}
which combined with Eq.~\eqref{eqn:1st_law_inequality} yields
\begin{equation}
    W^{AB}_{t} \geq \Delta F^{A}\;.
\end{equation}
These re-derivations of well known thermodynamical relations further justify the definitions of internal energy $\mathscr{U}^{A}$ and thermodynamic work $W^{AB}$ for arbitrary coupling strengths between systems $A$ and $B$ as used in \cite{esposito2010entropy}. 

\end{document}